\documentclass[3p, times, preprint]{elsarticle}
\usepackage{amsmath}
\usepackage{amssymb}
\usepackage{amsthm}
\usepackage{thmtools}
\usepackage{proof}
\usepackage[normalem]{ulem}
\usepackage{comment}
\usepackage{xcolor}
\definecolor{mydeepred}{HTML}{c51b1d}
\usepackage{tikz-cd}
\usepackage[hidelinks]{hyperref}

\newtheorem{theorem}{Theorem}
\newtheorem{lemma}[theorem]{Lemma}
\newtheorem{corollary}[theorem]{Corollary}
\newtheorem{claim}[theorem]{Claim}
\newtheorem*{claim*}{Claim}
\newtheorem{proposition}[theorem]{Proposition}
\newdefinition{definition2}[theorem]{Definition}
\newdefinition{notation}[theorem]{Notation}
\newdefinition{example}[theorem]{Example}

\usepackage{tikz}
\usetikzlibrary{arrows.meta, automata, positioning, decorations.pathmorphing, decorations.pathreplacing, shapes, calc}
\usepackage{lscape}
\usepackage{colortbl}
\usepackage{array}

\usepackage{macros}

\allowdisplaybreaks

\begin{document}

\begin{frontmatter}
\pdfoutput=1 

\bibliographystyle{plainurl}

\title{Regular Expressions with Backreferences on Multiple Context-Free Languages, and the Closed-Star Condition\tnoteref{t1}}

\author[aff1]{Taisei Nogami\corref{cor1}}
\ead{sora410@fuji.waseda.jp}

\author[aff1]{Tachio Terauchi\corref{cor1}}
\ead{terauchi@waseda.jp}
\ead[url]{https://terauchi.w.waseda.jp}

\affiliation[aff1]{organization={Waseda University}, addressline={3-4-1 Okubo, Shinjuku-ku}, city={Tokyo}, citysep={}, postcode={169-8555}, country={Japan}}
\cortext[cor1]{Corresponding authors.}

\begin{abstract}
    Due to miscellaneous practical extensions, regular expressions in the real world take on an entirely different look from the classical regular expressions found in the textbooks on formal languages and automata theory.  \emph{Backreference} is such a well known and practical extension of regular expressions and most modern programming languages, such as Java, Python, JavaScript and more, support regular expressions with backreferences (rewb) in their standard libraries for string processing.  A difficulty of backreference is non-regularity: unlike some other extensions, backreference strictly enhances the expressive power of regular expressions and thus rewbs can describe non-regular (in fact, even non-context-free) languages.  In this paper, we investigate the expressive power of rewbs by comparing rewbs to multiple context-free languages (MCFL) and parallel multiple context-free languages (PMCFL).  First, we prove that the language class of rewbs is a proper subclass of unary-PMCFLs, which coincide with the EDT0L languages.  Our result strictly improves the known (non-trivial) upper bound of rewbs, because the best-known bound was the intersection of the class of nondeterministic logspace languages and that of indexed languages, and, as we shall show in this paper, the class of EDT0L languages is a proper subclass of the intersection.  Additionally, we show that, however, the language class of rewbs is not contained in that of MCFLs even when restricted to rewbs with only one capturing group and no captured references.  Therefore, in general, the parallelism seems essential for rewbs.  Backed by these results, we define a novel syntactic condition on rewbs that we call \emph{closed-star} and observe that it provides an upper bound on the number of times a rewb references the same captured string.  The closed-star condition allows dispensing with the parallelism, that is, we prove that the language class of closed-star rewbs falls inside the class of unary-MCFLs, which is equivalent to that of EDT0L systems of finite index.  Furthermore, we show that the language class of closed-star rewbs also falls inside the class of nonerasing stack languages (NESL).
\end{abstract}


\end{frontmatter}

\section{Introduction}
\label{sec:intro}

\emph{Backreference}, an extension of regular expressions, is widely spread in the practical world.  Most modern regular expression engines, such as those in the standard libraries of Java, Python, JavaScript and more, support the backreference extension.  A difficulty of backreference is that it substantially increases the expressive power of regular expressions.  Let $\alpha$ be the rewb $(_1 (a+b)^\ast )_1\,\bs 1$ over the alphabet $\syuugou{a,b}$, for instance.  The language $L(\alpha)$ of $\alpha$ is the copy language $\syuugou[ww]{w \in \syuugou{a,b}^\ast}$, a typical example of a non-context-free language.  Intuitively, the subexpression $\bs 1$ matches only the string matched by $(a+b)^\ast$, the subexpression in the preceding $1$-indexed capturing group.  We rephrase this behavior as follows: $\bs 1$ \emph{backreferences} the string captured in $(_1 (a+b)^\ast )_1$.  In this manner, backreference can assert equality constraints on the substrings of a matching string.  
Thus, we reach the following natural question: how powerful is backreference?  In prior study, C\^{a}mpeanu et al.~\cite{campeanu2003formal}, and Berglund and van der Merwe~\cite{berglund2023re} showed that the language class of rewbs is a proper subclass of context-sensitive languages (CSL), but is incomparable with the class of context-free languages (CFL).  Our recent study~\cite{mfcs2023,ic2025} showed that the language class of rewbs is a proper subclass of indexed languages (IL) but not a subclass of stack languages (SL), which was proposed by Ginsburg et al.~\cite{ginsburg1967stack, ginsburg1967one} and properly sits between the classes of CFLs and ILs~\cite{aho1969nested, greibach1969checking}.  Note that Aho, who invented ILs, has shown that $\text{CFL} \subsetneq \text{IL} \subsetneq \text{CSL}$~\cite{aho1968indexed}. As recently noted by Uezato~\cite{uezato2024regular}, the result of Schmid~\cite{schmid2013inside} implies that NL, the class of languages that are decidable by nondeterministic logspace Turing machines, contains the language class of rewbs.  Also, Uezato has observed that $\text{NL} \nsubseteq \text{IL}$, and $\text{IL} \nsubseteq \text{NL}$ unless $\text{NL} = \text{NP}$.  Therefore, to our knowledge, the best-known non-trivial upper bound for rewb languages is $\text{NL} \cap \text{IL}$.\footnote{By \emph{non-trivial} upper bound, we exclude trivial upper bounds such as the language class of rewbs itself and those of extensions of rewbs.}

We investigate the expressive power of rewbs relative to those of multiple context-free languages (MCFL) and parallel multiple context-free languages (PMCFL) proposed by Seki et al.~\cite{seki1991multiple} as the language classes of multiple context-free grammars (MCFG) and parallel multiple context-free grammars (PMCFG), respectively.
As for the expressive powers of MCFG and PMCFG, Seki et al.~\cite{seki1991multiple} showed the inclusions $\text{CFL} \subsetneq \text{MCFL} \subsetneq \text{PMCFL} \subsetneq \text{CSL}$.  Moreover, Nishida and Seki~\cite{nishida2000grouped} showed that a special subclass of PMCFLs, unary-PMCFLs, is equivalent to the language class of EDT0L systems.\footnote{See \cite{rozenberg1973extension,engelfriet1978tree} for the definition of EDT0L systems.}

We first show that the class of unary-PMCFLs, and hence that of EDT0L languages, contains the language class of rewbs by constructing a unary-PMCFG $G_{\alpha}$ equivalent to an arbitrary given rewb $\alpha$.
This result strictly improves the best-known non-trivial upper bound for rewb languages as follows.  Note that $\text{EDT0L} \subseteq \text{NL} \cap \text{IL}$ is known~\cite{jones1976recognition,salomaa1974parallelism}, while we show in Section~\ref{sec:rewbpmcfg} that a CFL known to lie outside EDT0L~\cite{ehrenfeucht1977some} belongs to NL (see Proposition~\ref{prop:b1nl}).  Thus, $\text{EDT0L} \subsetneq \text{NL} \cap \text{IL}$, and our result strictly improves the best-known non-trivial upper bound for rewb languages from $\text{NL} \cap \text{IL}$ as stated above to $\text{EDT0L}$.
Additionally, the CFL makes \text{EDT0L} the first non-trivial upper bound of rewbs that does not contain CFLs.

Notable in our construction is that the parallelism of PMCFG plays a vital role.  In fact, we show a simple rewb with only one capturing group and no captured references\footnote{A \emph{captured reference} is, as its name would suggest, a reference $\bs i$ inside a capturing group $(_j\;)_j$~\cite{mfcs2023,ic2025}} that describes a non-MCFL using a pumping lemma for MCFL.  Therefore, the language class of rewbs and the class of MCFLs are incomparable.
The fact that this rewb has no captured reference is significant because, as shown by \cite{mfcs2023,ic2025}, the language class of rewbs without captured references is contained in that of nonerasing stack languages (NESL), a proper subclass of SLs~\cite{ogden1969intercalation,greibach1969checking}, while not every rewb language belongs to SL as mentioned earlier.

Furthermore, motivated by these results, we propose a novel syntactic variant of rewbs that we call \emph{closed-star} rewbs.  Intuitively, being closed-star means that no starred subexpression can backreference strings captured outside the star.  For example, $((_1 (a+b)^\ast )_1\, c\,\bs 1)^\ast$ is closed-star, whereas a simple example of a non-closed-star rewb is $(_1 a^\ast )_1 (c \bs 1)^\ast$ which also happens to describe a non-MCFL (cf.~Theorem~\ref{thm:notmcfg}).  As we shall show, a closed-star rewb has the property that we can compute an upper bound on the number of times the rewb references the same captured string.  We use this property of the closed-star condition to prove that the language class of rewbs is properly contained in that of unary-MCFGs, which is known to coincide with that of EDT0L systems of finite index~\cite{kanazawa2010copying}. 
Moreover, we also show that every closed-star rewb describes an NESL.  
Below, we list the main contributions of the paper:
\begin{enumerate}[(a)]
    \item Every rewb describes a unary-PMCFL (Theorem~\ref{thm:rewbpmcfg}).
    \item There exists a rewb with only one capturing group and no captured references that describes a non-MCFL (Theorem~\ref{thm:notmcfg}).
    \item A novel syntactic variant of closed-star rewbs (Definition~\ref{def:sc}).
    \item Every closed-star rewb describes a unary-MCFL (Theorem~\ref{thm:rewbscmcfg}).
    \item Every closed-star rewb describes an NESL (Theorem~\ref{thm:rewbscisnesa}).
\end{enumerate}
Also, as a corollary, it follows by (b) and (d) that the language class of closed-star rewbs is a proper subclass of rewbs (Corollary~\ref{cor:scsmaller}). 
As remarked above, we have recently shown that the language class of rewbs is not contained in SL and that of rewbs with no captured references is a proper subclass of NESLs~\cite{mfcs2023,ic2025}.  From these results, it follows by (a) and (b) that unary-PMCFL is not contained in SL and NESL is not contained in MCFL, respectively (Corollary~\ref{cor:pmcfgsaandmcfgnesa}).
The rest of the paper is organized as follows: Section~\ref{sec:prelim} presents formal definitions of rewbs, MCFLs and PMCFLs.  In Sections~\ref{sec:rewbpmcfg}, \ref{sec:rewbscmcfg} and \ref{sec:rewbscnesa}, we prove the main contributions listed above. Section~\ref{sec:related} discusses related work.  Finally, Section~\ref{sec:conc} concludes the paper.
\ref{app:table} provides a summary table of the language inclusion results.

\section{Preliminaries}
\label{sec:prelim}

We first formalize the syntax and semantics of rewbs following the prior study~\cite{schmid2016characterising,freydenberger2019deterministic,mfcs2023,ic2025}.
Let $\Sigma$ denote an alphabet and $\mynat$ denote the set of all positive integers.
\begin{definition2} \label{def:rewb}
    Define the set of \emph{rewbs} by the following grammar:
        \[
            \alpha ::= a \mid \varepsilon \mid \bs i \mid \alpha_1 \alpha_2 \mid \alpha_1 + \alpha_2 \mid \alpha^\ast \mid (_i \alpha )_i
        \]
    where, as for the last rule, $\alpha$ appearing in $(_i \alpha )_i$ can have neither the same indexed capturing group $(_i \; )_i$ nor reference $\bs i$ as a subexpression.
\end{definition2}

For example, none of $(_1 a^\ast \bs 1 )_1$, $(_1 (_1 a^\ast )_1 )_1$ and $(_1 (_2 (_1 a^\ast )_1 )_2 )_1$ are a valid expression.  
Note that we adopt the ``may repeat labels'' convention~\cite{berglund2023re}, i.e., the same indexed capturing group can appear more than once in an expression, like $(_1 a^\ast )_1 (\bs 1 (_1 b^\ast )_1 ) ^ \ast$.

Hereafter, we define the \emph{ref-language} of a rewb and the \emph{dereferencing function} for defining the formal semantics.  Roughly, the semantics first regards a rewb over $\Sigma$ as a pure regular expression over the extended alphabet $\Sigma \mydisjointu \brkset \mydisjointu \mynat$, where $\brkset = \syuugou[\lbrack_i, \rbrack_i]{i \in \mynat}$ and $\mydisjointu$ is a disjoint union, whose (regular) language is called the ref-language of the rewb.  Then, the semantics defines the language of the rewb as the image of its ref-language under the dereferencing function.
\begin{definition2}
    Let $\alpha$ be a rewb.  Define the \emph{ref-language} $\reflang(\alpha)$ of $\alpha$ inductively as follows:
        \begin{gather*}
            \reflang(a) = \syuugou{a}, \reflang(\varepsilon) = \syuugou{\varepsilon}, \reflang(\bs i) = \syuugou{i}, \reflang(\alpha_1 \alpha_2) = \reflang(\alpha_1) \reflang(\alpha_2), \\
            \reflang(\alpha_1 + \alpha_2) = \reflang(\alpha_1) \cup \reflang(\alpha_2), \reflang(\alpha^\ast) = \reflang(\alpha)^\ast, \reflang((_i \alpha )_i) = \syuugou{\lbrack_i}\reflang(\alpha)\syuugou{\rbrack_i}
        \end{gather*}
    where $a \in \Sigma$ and $i \in \mynat$.
\end{definition2}

Next, we explain the dereferencing function $\deref$.  In essence, for each number character appearing in the argument string, $\deref$ sequentially replaces the number character with the string bracketed by the same number.  Let us give examples illustrating how $\deref$ works.
\begin{enumerate}
    \item $\lbrack_1 a \lbrack_2 b \rbrack_2\,2 \rbrack_1\,1$.  $\deref$ first encounters $2$, then goes back to fetch the $2$-bracketed string $b$ from the prefix string and replaces the $2$ with $b$.  Thus, we have the temporary string $\lbrack_1 a \lbrack_2 b \rbrack_2 b \rbrack_1 1$.  $\deref$ again processes the same replacement on this string.  Namely, $\deref$ encounters $1$, fetches $abb$ (but not $a \lbrack_2 b \rbrack_2 b$) and replaces the $1$ with $abb$.  Note that $\deref$ fetches the bracketed string but with inner bracket symbols omitted.  Now, we have $\lbrack_1 a \lbrack_2 b \rbrack_2 b \rbrack_1 a b b$ and this string has no number character.  Therefore, the value of $\deref$ is $abbabb$ by omitting bracket symbols.  Here is the diagram: $\lbrack_1 a\,\underline{\lbrack_2 b \rbrack_2}\, 2\,\rbrack_1\,1 \to \underline{\lbrack_1 a\,\lbrack_2 b \rbrack_2\, b\,\rbrack_1}\,1 \to \lbrack_1 a\,\lbrack_2 b \rbrack_2\, b\,\rbrack_1\,abb \to abbabb$.
    \item $\lbrack_1 a \rbrack_1\,1 \lbrack_1 bb \rbrack_1\,1$.  The diagram is $\underline{\lbrack_1 a \rbrack_1}\,1\,\lbrack_1 bb \rbrack_1\,1 \to \lbrack_1 a \rbrack_1\,a\,\underline{\lbrack_1 bb \rbrack_1}\,1 \to \lbrack_1 a \rbrack_1\,a\,\lbrack_1 bb \rbrack_1\,bb \to aabbbb$.  In the second step, $\deref$ fetches the rightmost bracketed string $bb$ for a number character if the same numbered brackets appear in multiple positions before the number character.
    \item $abc\,1\,2$.  The diagram is $abc\,1\,2 \to abc\, 2 \to abc$.  Namely, $\deref$ fetches $\varepsilon$ for a number character if there is no corresponding brackets before the number character.
\end{enumerate}

Formally, we define $\deref$ as follows.  For a character $c$ and a string $s$, we write $c \in s$ to mean that $c$ appears in $s$.
\begin{definition2} \label{def:deref}
    Let $i \in \mynat$. 
    Let $g: (\Sigma \mydisjointu \brkset)^\ast \to \Sigma^\ast$ denote the free monoid homomorphism where $g(a) = a$ for each $a \in \Sigma$ and $g(b) = \varepsilon$ for each $b \in \brkset$.
    \begin{enumerate}[(a)]
        \item Define the \emph{fetching function} $\myfet_i: (\Sigma \mydisjointu \brkset)^\ast \to \Sigma^\ast$ that maps $v$ as follows:
            (1) if $v$ decomposes as 
            $x_1 \lbrack_i x_2$ ($\lbrack_i \notin x_2$) and $x_2 = y_1 \rbrack_i y_2$ ($\rbrack_i \notin y_1$)
            , then $g(y_1)$, (2) if $v$ decomposes as $x_1 \lbrack_i x_2$ ($\lbrack_i \notin x_2$) and $\rbrack_i \notin x_2$, then $g(x_2)$, and (3) if $\lbrack_i \notin v$, then $\varepsilon$.\footnote{In fact, case (2) never happens when the argument string is from a ref-language.}
        \item Define the \emph{pre-dereferencing function} $\deref^{\circ}: (\Sigma \mydisjointu \brkset \mydisjointu \mynat)^\ast \to (\Sigma \mydisjointu \brkset)^\ast$ as follows: $\deref^{\circ}(\varepsilon) = \varepsilon$, $\deref^{\circ}(va) = \deref^{\circ}(v)a$, $\deref^{\circ}(vb) = \deref^{\circ}(v)b$ and $\deref^{\circ}(vi) = \deref^{\circ}(v) \myfet_{i}(\deref^{\circ}(v))$ where $a \in \Sigma$ and $b \in \brkset$. 
        \item Define the \emph{bracketed-string-projection function} $\mymem_i$ as $\myfet_{i} \cdot \deref^{\circ}$ where the interpunct $\cdot$ is the function composition.  Hence, the last rule above can be written as $\deref^{\circ}(vi) = \deref^{\circ}(v) \mymem_i(v)$.
        \item Define the \emph{dereferencing function} $\deref$ as $g \cdot \deref^{\circ}$.
    \end{enumerate}
\end{definition2}

Note that $\mymem_i(v)$ represents the rightmost $i$-bracketed or semi-$i$-bracketed string (i.e., that without the closing bracket, e.g., $\mymem_1(\lbrack_1 ab) = ab$) as of when $\deref$ processed $v$.  Later, we present an alternative inductive characterization of $\mymem_i$ in Proposition~\ref{prop:memind}.
Finally, the language $L(\alpha)$ of a rewb $\alpha$ is defined as $\deref(\reflang(\alpha))$.  We give an example of calculating a rewb language.

\begin{example}
  \label{ex:copylang}
    Recall the rewb $\alpha = (_1 (a+b)^\ast )_1\,\bs 1$ from the introduction.  The ref-language $\reflang(\alpha)$ is $\syuugou[\lbrack_1 w \rbrack_1\,1]{w \in \syuugou{a,b}^\ast}$.  For each $\lbrack_1 w \rbrack_1\,1 \in \reflang(\alpha)$, $\deref$ processes it as $\underline{\lbrack_1 w \rbrack_1}\,1 \to \lbrack_1 w \rbrack_1 w \to ww$.  Therefore, we conclude $L(\alpha) = \deref(\reflang(\alpha)) = \syuugou[ww]{w \in \syuugou{a,b}^\ast}$, that is, the copy language. 
\end{example}

Next, we recall MCFG and PMCFG~\cite{seki1991multiple}.  
The most outstanding feature of the grammars is that a nonterminal can derive a tuple of strings instead of just a single string.  Each production rule is of the form $A \to f[A_1, \dots, A_l]$ and means that the nonterminal $A$ on the left-hand side can derive string tuples formed by mapping $l$ string tuples, each derived from the nonterminals $A_1, \dots, A_l$ on the right-hand side, with the string concatenation function $f: (\Sigma^\ast)^{d(A_1)} \times \cdots \times (\Sigma^\ast)^{d(A_l)} \to (\Sigma^\ast)^{d(A)}$.  Here, the dimension of tuples derived from $A$ is unified with the designated constant $d(A)$.  In particular, $d(S) = 1$ is required to finally derive strings but not string tuples.\footnote{We identify a string with a unidimensional tuple. \label{fn:unitplisstr}}
MCFGs lack the parallelism of PMCFGs in the sense that none of its string concatenation functions can use the same argument string more than once.
The formal syntax and semantics are as follows:
\begin{definition2}
    Let $m \in \mynat$.  An \emph{$m$-parallel multiple context-free grammar} ($m$-PMCFG) $G$ is a 5-tuple $(\mynonterminals,\Sigma,\mathcal{F},P,S)$, where $\mynonterminals$ is a finite set of nonterminals, $\Sigma$ is an alphabet, $\mathcal{F}$ is a finite set of functions, $P$ is a finite set of rules, and $S \in \mynonterminals$ is an initial nonterminal.
    Additionally, $G$ must satisfy the following conditions:
    \begin{itemize}
        \item A function $f \in \mathcal{F}$ has the type $(\Sigma^\ast)^{d_1} \times \dots \times (\Sigma^\ast)^{d_l} \to (\Sigma^\ast)^{d}$ for some integers $l \geq 0$, $d_i \geq 1 (1 \leq i \leq l)$ and $d \geq 1$ that depend on $f$.
            Note that we allow $l = 0$ (i.e., $f$ can be nullary) and we identify a nullary $f$ with a tuple of dimension $d$.
            Moreover, $f$ maps
            $x_{i,j}$'s ($1 \leq i \leq l$ and $1 \leq j \leq d_i$) into $(c_{k,0} x'_{k,1} c_{k,1} \cdots x'_{k,h_k} c_{k,h_k})$ ($1 \leq k \leq d$), where (1) $h_k$ is some integer depending on $f$, and (2) $x'_{k,r}$ is some $x_{i,j}$ and $c_{k,r} \in \Sigma^\ast$ for $0 \leq r \leq h_k$.
        \item A nonterminal $A \in \mynonterminals$ has an integer $1 \leq d(A) \leq m$ called the \emph{dimension} of $A$.  In particular, $d(S) = 1$.  
        \item Let $f \in \mathcal{F}$ and fix $l$, $d_1, \dots, d_l$ and $d$ as mentioned above.              
            %
            %
            A rule is of the form $A \to f[A_1, \dots, A_l]$ where $d = d(A)$ and $d_i = d(A_i)$ for each $1 \leq i \leq l$.  
    \end{itemize}
    In particular, an \emph{$m$-multiple context-free grammar} ($m$-MCFG) is an $m$-PMCFG with the restriction that no $f \in \mathcal{F}$ can use the same $x_{i,j}$ more than once (i.e.,~no functions are parallel).
    An $m$-PMCFG (resp.~$m$-MCFG) is called \emph{unary} if every function is nullary or unary (i.e.,~$l \leq 1$)~\cite{nishida2000grouped}.
\end{definition2}
\begin{definition2}
    Let $m \in \mynat$ and $G$ be an $m$-PMCFG or $m$-MCFG.  For each nonterminal $A$, define the set $L_G(A)$ of all tuples derived from $A$ as the smallest set satisfying the following condition:
    $f(\xi_1, \dots, \xi_l) \in L_G(A)$ for each rule $A \to f[A_1,\dots,A_l]$ and for every $\xi_i \in L_G(A_i)$ ($1 \leq i \leq l$).
    Define the \emph{language} $L(G)$ of $G$ as $L_G(S)$, where $S$ is the initial nonterminal of $G$.  
    The language classes of $m$-PMCFGs and $m$-MCFGs are called \emph{$m$-parallel multiple context-free languages} ($m$-PMCFL) and \emph{$m$-multiple context-free languages} ($m$-MCFL), respectively.
\end{definition2}

We denote the union of all $m$-PMCFGs over all $m \in \mynat$ simply as PMCFGs, and similarly for MCFGs, PMCFLs and MCFLs.  Because we only need unary PMCFGs and MCFGs except in Theorem~\ref{thm:notmcfg}, we introduce following short-hand notations.  We use $\lambda$-notation to denote a unary function $f: (\Sigma^\ast)^{d_1} \to (\Sigma^\ast)^d$ as $\lambda (x_1, \dots, x_{d_1}). (y_1, \dots, y_{d})$.  For example, $\lambda (x,y). xy$ is a unary (but not binary) function with $l = 1, d_1 = 2$ and $d = 1$ (recall footnote~\ref{fn:unitplisstr}).  Also, we simply write $A \to f\,A_1$ for a rule $A \to f[A_1]$ and a unary function $f$.
We borrow examples of a PMCFG and an MCFG from~\cite{seki1991multiple}.
\begin{example}    
    Let $\Sigma$ denote $\syuugou{a,b}$ and $G = (\syuugou{S,A},\Sigma,\mathcal{F},P,S)$ be a $1$-PMCFG, where $\mathcal{F}$ consists of $\mycp = \lambda x. xx$, $\myapp_{a} = \lambda x. xa$ and $\myapp_{b} = \lambda x. xb$, and $P = \{\, S \to \mycp\,A, A \to \myapp_{a}\,A \mid \myapp_{b}\,A \mid \varepsilon \,\}$.\footnote{$A \to \gamma_1 \mid \gamma_2 \mid \cdots$ is a shorthand for $A \to \gamma_1, A \to \gamma_2, \dots$.}
    %
    Note that $\mycp$ is parallel.
    $G$ derives $abab$ by 
    \[
        S \to \mycp\,A \to \mycp\cdot\myapp_{b}\,A \to \mycp\cdot\myapp_{b}\cdot\myapp_{a}\,A \to \mycp\cdot\myapp_{b}\cdot\myapp_{a}\,\varepsilon = \mycp\cdot\myapp_{b}\,a = \mycp\,(ab)  = abab.
    \]
    We can easily check that $L_G(A) = \syuugou{a,b}^\ast$ and $L(G) = \syuugou[ww]{w \in \syuugou{a,b}^\ast}$, that is, the copy language from the introduction and Example~\ref{ex:copylang}.
    
    Because $1$-MCFL coincides with CFL~\cite{seki1991multiple}, no $1$-MCFG can describe $L(G)$. However, a $2$-MCFG can.  Let $G' = (\syuugou{S,A},\Sigma,\mathcal{F}',P',S)$ be a $2$-MCFG, where $\mathcal{F}'$ consists of $\mycat = \lambda (x,y). xy$, $\myapp'_{a} = \lambda (x,y). (xa,ya)$ and $\myapp'_{b} = \lambda (x,y). (xb,yb)$, and $P' = \{\,S \to \mycat\,A, A \to \myapp'_{a}\,A \mid \myapp'_{b}\,A \mid (\varepsilon,\varepsilon)\,\}$.
    %
    Then, $L_{G'}(A) = \syuugou[(w,w)]{w \in \syuugou{a,b}^\ast}$ and $L(G') = L(G)$.
    For example, $G'$ derives $abab$ by 
    \[
        S \to \mycat\,A \to \mycat\cdot\myapp'_{b}\,A \to \mycat\cdot\myapp'_{b}\cdot\myapp'_{a}\,A \to \mycat\cdot\myapp'_{b}\cdot \myapp'_{a}\,(\varepsilon,\varepsilon) = \mycat\,(ab,ab) = abab.
    \]
\end{example}

\section{Every rewb describes a unary-PMCFL but not necessarily an MCFL}
\label{sec:rewbpmcfg}

First, we prove the following theorem. As noted in the introduction, together with prior work and Proposition~\ref{prop:b1nl} proved at the end of this section, this strictly improves the best-known non-trivial upper bound of rewbs (i.e.,~$\text{NL} \cap \text{IL}$).
\begin{restatable}{theorem}{thmrewbpmcfg}
  \label{thm:rewbpmcfg}
    Every rewb describes a unary-PMCFL.
\end{restatable}

We informally explain the proof.  For an arbitrary given rewb $\alpha$ using number characters each of whose value is bounded by $\kbound$, we construct a $(\kbound+1)$-unary-PMCFG $G_{\alpha}$ as follows.  In a $\kbound+1$ dimension tuple $(x,y_1,\dots,y_\kbound)$, we use the first component $x$ to store the generating string and the other $\kbound$ components $y_1, \dots, y_\kbound$ to keep the contents of $\kbound$ \emph{memory cells}.
A memory cell $k$ is a holder for the dereferencing process to place the string most recently captured by number character $k$.\footnote{We adopt the terminology from a class of automata equivalent to rewbs called memory automata~\cite{schmid2016characterising}, which are finite automata augmented with \emph{memories} that store the most recently captured string for each capturing group index.}
Let us fix an NFA $\myautomaton_{\alpha}$ equivalent to the ref-language $\reflang(\alpha)$.
Each nonterminal $A_i$ in $G_{\alpha}$ corresponds to a state $q_i$ in $\myautomaton_{\alpha}$ and each rule corresponds to a transition in $\myautomaton_{\alpha}$.
%
%
%
$G_{\alpha}$ guesses a run of $\myautomaton_\alpha$ accepting, say $v \in L(\myautomaton_\alpha) = \reflang(\alpha)$, backwardly while changing $A_j$ to $A_i$ but with the function corresponding to $c$ applied for each transition $q_i \myto{c}{} q_j$.
Some functions have a superscript $M \subseteq \{ 1, \dots, \kappa \}$ that denotes the set of open memory cells, where an open memory cell is, roughly, one whose content is currently being defined.
If the character being scanned is $a \in \Sigma$, then $G_{\alpha}$ applies $\symi_{a}^{M}$ (\emph{input}) which appends $a$ to $x$ and $y_k$ for all $k \in M$.  
If it is $i \in \mynat$, then $G_{\alpha}$ applies $\symp_{i}^{M}$ (\emph{paste}) which appends $y_i$ to $x$ and $y_k$ for all $k \in M$.
If it is $\lbrack_i$, then $G_{\alpha}$ applies $\symr_{i}$ (\emph{reset}) which empties $y_i$.  
If it is $\rbrack_i$, then $G_{\alpha}$ applies $\symid$ which is the identity function.
$G_{\alpha}$ ends the scan by applying a nullary function of the form $(\varepsilon,\varepsilon,\dots,\varepsilon)$.  
Therefore, the sequence of (non-nullary) functions are applied to $(\varepsilon,\varepsilon,\dots,\varepsilon)$ in the same order as $v$ to return a tuple whose first component is $\deref(v)$, the dereferenced string of $v$.
We call this part the \emph{functional dereferencing phase}.
Finally, $G_{\alpha}$ extracts this $\deref(v)$ by $\symo$ (\emph{output}) applied at the beginning of the derivation.
We give an example of the proof construction.  The formal proof is coming later.
\begin{example}
    \label{ex:pmcfgproof}
    Let $\Sigma = \syuugou{a,c}$, $\alpha = (_1 a^\ast )_1\,(c \bs 1)^\ast$ and $\kbound = 1$.  We present how to construct the $2$-unary-PMCFG $G_\alpha$. 
    The following diagram shows an NFA $\myautomaton_\alpha$ that recognizes $\reflang(\alpha)$:
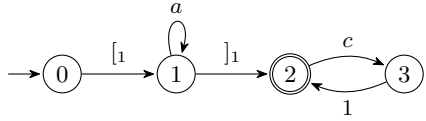
    \begin{figure}[htb]
        \centering
        \begin{tikzpicture}[
                shorten >=1pt,
                node distance=2.25cm,
                on grid,
                >={Stealth[round]},
                auto
            ]
            \tikzset{el/.style = {inner sep=2pt, align=left, sloped}}
            \tikzset{every node/.style={font=\small}}
            \tikzset{every loop/.style={min distance=6mm,looseness=8}}
            \tikzset{every state/.style={inner sep=.1mm, minimum size=5mm}}
            \tikzset{initial text=}
            \tikzset{every initial by arrow/.style={->}}
            \tikzset{bend angle=30}

            \node[state,initial] (0) at (0,-3) {$0$};
            \node[state] (1) at (1.5,-3) {$1$};
            \node[state,accepting] (2) at (3,-3) {$2$};
            \node[state] (3) at (4.5,-3) {$3$};
            
            \path[->] (0) edge node {\footnotesize
                            $\lbrack_1$
                        } (1);
            \path[->] (1) edge [in=10,out=170,loop above] node {\footnotesize
                            $a$
                        } ();
            \path[->] (1) edge node {\footnotesize
                            $\rbrack_1$
                        } (2);
            \path[->] (2) edge [bend left=24] node {\footnotesize
                            $c$
                        } (3);
            \path[->] (3) edge [bend left=24] node {\footnotesize
                            $1$
                        } (2);
        \end{tikzpicture}
    \end{figure}

    Let $G_\alpha$ be the $2$-unary-PMCFG $(\syuugou{S, A_0, A_1, A_2, A_3}, \Sigma, \mathcal{F}, P, S)$, where $\mathcal{F}$ consists of 
        \begin{gather*}
            \symo = \lambda (x,y_1). x, \quad 
            \symi_{a}^{\syuugou{1}} = \lambda (x,y_1). (xa,y_1a), \quad
            \symi_{c}^{\syuugou{}} = \lambda (x,y_1). (xc,y_1), \\
            \symp_{1}^{\syuugou{}} = \lambda (x,y_1). (x y_1, y_1), \quad
            \symr_{1} = \lambda (x,y). (x,\varepsilon) \text{\quad and \quad} 
            \symid = \lambda (x,y). (x,y),
        \end{gather*}
    and $P$ consists of $S \to \symo A_2$, $A_3 \to \symi_{c}^{\syuugou{}} A_2$, $A_2 \to \symp_{1}^{\syuugou{}} A_3 \!\mid\! \symid A_1$, $A_1 \to \symi_{a}^{\syuugou{1}} A_1 \!\mid\! \symr_{1} A_0$ and $A_0 \to (\varepsilon,\varepsilon)$.
    %
    For instance, $G_{\alpha}$ derives $aca$ by 
        \begin{align*}
            S &\to 
            \symo A_2 \to 
            \symo\cdot\symp_{1}^{\syuugou{}} A_3 \to 
            \symo\cdot\symp_{1}^{\syuugou{}}\cdot\symi_{c}^{\syuugou{}} A_2 \to 
            \symo\cdot\symp_{1}^{\syuugou{}}\cdot\symi_{c}^{\syuugou{}}\cdot\symid A_1 \to 
            \symo\cdot\symp_{1}^{\syuugou{}}\cdot\symi_{c}^{\syuugou{}}\cdot\symid\cdot\symi_{a}^{\syuugou{1}} A_1 \\
              & \to
            \symo\cdot\symp_{1}^{\syuugou{}}\cdot\symi_{c}^{\syuugou{}}\cdot\symid\cdot\symi_{a}^{\syuugou{1}}\cdot\symr_1 A_0 \to
            \symo\cdot\symp_{1}^{\syuugou{}}\cdot\symi_{c}^{\syuugou{}}\cdot\symid\cdot\symi_{a}^{\syuugou{1}}\cdot\symr_1 (\varepsilon, \varepsilon) \\
             &= \symo\cdot\symp_{1}^{\syuugou{}}\cdot\symi_{c}^{\syuugou{}}\cdot\symid\cdot\symi_{a}^{\syuugou{1}} (\varepsilon, \varepsilon)
             = \symo\cdot\symp_{1}^{\syuugou{}}\cdot\symi_{c}^{\syuugou{}}\cdot\symid (a, a)
             = \symo\cdot\symp_{1}^{\syuugou{}}\cdot\symi_{c}^{\syuugou{}} (a, a)
             = \symo\cdot\symp_{1}^{\syuugou{}} (ac, a)
             = \symo (aca, a)
             = aca.
        \end{align*}
    The last equality chain is the functional dereferencing phase.
\end{example}

Later in this section, we shall show that $L(\alpha)$ in this example is not an MCFL.
Hence, the parallelism of $\symp_{1}^{\syuugou{}} = \lambda (x,y_1). (x y_1, y_1)$ is indispensable for the unary-PMCFG $G_{\alpha}$.
We introduce a notation for calculating open memory cells.
\begin{definition2} The mapping $\myopen: (\Sigma \mydisjointu \brkset \mydisjointu \mynat)^\ast \to \powerset(\mynat)$ is defined inductively as follows: $\myopen(\varepsilon) = \emptyset$, $\myopen(va) = \myopen(vi) = \myopen(v)$, $\myopen(v\lbrack_i) = \myopen(v) \cup \syuugou{i}$ and $\myopen(v\rbrack_i) = \myopen(v) \backslash \syuugou{i}$, where $a \in \Sigma$ and $i \in \mynat$.
\end{definition2}

For example, $\myopen(\lbrack_1 ab \lbrack_2 c) = \syuugou{1,2}$ and $\myopen(\lbrack_1 ab \rbrack_1 \lbrack_2 c) = \syuugou{2}$.  
The following proposition is immediate.
\begin{proposition} \label{prop:op}
    Let $v$ be a string over $\Sigma \mydisjointu \brkset \mydisjointu \mynat$.  For all $i \in \mynat$, $i \in \myopen(v)$ if and only if there exist $x,y \in (\Sigma \mydisjointu \brkset \mydisjointu \mynat)^\ast$ such that $v = x\mathord{\lbrack_i}y$ and $\mathord{\lbrack_i},\mathord{\rbrack_i} \notin y$.
\end{proposition}

Recall that $\mymem_j(v) = \myfet_{j} \cdot \deref^{\circ}$ represents the $j$-bracketed or semi-$j$-bracketed string in the dereference of $v$ (cf.~Definition~\ref{def:deref}).  For example, $\mymem_1(\lbrack_1 ab \rbrack_1 c) = ab$, $\mymem_1(\lbrack_1 ab ) = ab$, $\mymem_1(\lbrack_1 ab \rbrack_1 \lbrack_1) = \varepsilon$ and $\mymem_2(\lbrack_1 ab \rbrack_1 \lbrack_2 1 c) = abc$.
We present an inductive characterization for $\mymem_j(v)$ using $\myopen$.
\begin{proposition} \label{prop:memind}
    Let $a \in \Sigma$ and $i,j \in \mynat$.
    Then, $(\mathrm{a})$ $\mymem_j(va) = \mymem_j(v) a$ if $j \in \myopen(v)$, otherwise $\mymem_j(va) = \mymem_j(v)$, $(\mathrm{b})$ $\mymem_j(vi) = \mymem_j(v) \mymem_i(v)$ if $j \in \myopen(v)$, otherwise $\mymem_j(vi) = \mymem_j(v)$, $(\mathrm{c})$ $\mymem_i(v\lbrack_i) = \varepsilon$ and $\mymem_j(v\lbrack_i) = \mymem_j(v)$ for $j \neq i$, and $(\mathrm{d})$ $\mymem_j(v\rbrack_i) = \mymem_j(v)$.
\end{proposition}
\begin{proof}
    The former case of (c) can be easily checked from the definition of $\mymem$.  
    %
    For the remaining cases, it suffices to prove the following claim.
    \begin{claim*}
    Let $u$ be in $(\Sigma \mydisjointu (\brkset \backslash \syuugou{\lbrack_j,\rbrack_j}))^\ast$ or $\rbrack_j$.
    If $j \in \myopen(v)$, then $\mymem_j(vu) = \mymem_j(v) g(u)$.  Otherwise, $\mymem_j(vu) = \mymem_j(v)$. 
    \end{claim*}
    \begin{proof}
        Suppose that $j \in \myopen(v)$.  By Proposition~\ref{prop:op}, we have the decomposition $v = x \lbrack_j y$ where $\mathord{\lbrack_j},\mathord{\rbrack_j} \notin y$.  
        From the definition, $\deref^{\circ}(x \lbrack_j y)$ can be written as $x' \lbrack_j y'$ where $\mathord{\lbrack_j}, \mathord{\rbrack_j} \notin y'$. 
    Hence, we have $\mymem_j(vu) = \myfet_j( \deref^{\circ}(vu) ) = \myfet_j( \deref^{\circ}(v) u ) = \myfet_j( x' \lbrack_j y' u ) = g(y') g(u)$ and
    $\mymem_j(v) = \myfet_j( \deref^{\circ}(v) ) = \myfet_j( x' \lbrack_j y') = g(y')$.  
    Suppose that $j \notin \myopen(v)$.  
        If $\mathord{\lbrack_j} \notin v$, then $\deref^{\circ}(v)$ also has no $\mathord{\lbrack_j}$ and thus $\mymem_j(vu) = \mymem_j(v) = \varepsilon$.  
        Otherwise (i.e., $\mathord{\lbrack_j} \in v$), $v$ can be written as $x \lbrack_j y \rbrack_j z$ where $y$ has neither $\mathord{\lbrack_j}$ nor $\mathord{\rbrack_j}$ and $z$ has no $\mathord{\lbrack_j}$ by Proposition~\ref{prop:op}.  We can write $\deref^{\circ}(x \lbrack_j y \rbrack_j z) = x' \lbrack_j y' \rbrack_j z'$ where $y'$ also has neither $\mathord{\lbrack_j}$ nor $\mathord{\rbrack_j}$.  
    We can easily check that $\mymem_j(vu) = \mymem_j(v) = g(y')$.
    \end{proof}

    Note that for case (b), $\deref^{\circ}(v \mymem_i(v)) = \deref^{\circ}(v) \mymem_i(v)$ by the definition of $\deref^{\circ}$ and thus $\mymem_j(vi) = \mymem_j(v\mymem_i(v))$ by the definition of $\mymem$.
\end{proof}

Next, we formally review NFA.  Here is the standard definition~\cite{hopcroft2001introduction}:
\begin{definition2} \label{def:nfa}
    A \emph{nondeterministic finite automaton} (NFA) is a $5$-tuple $(Q,\Sigma,\delta,q_0,F)$ where $Q$ is a finite set of states, $\Sigma$ is an alphabet, $\delta: Q \times \Sigma \to \powerset(Q)$ is a transition function, $q_0 \in Q$ is a start state and $F \subseteq Q$ is a set of final states.
\end{definition2}

Note that $\varepsilon$-transitions are not allowed.  As standard, we extend $\delta$ to $\hat{\delta}:Q\times \Sigma^\ast \to \powerset(Q)$ where $\hat{\delta}(q,w)$ is the set of all states reachable from $q$ via $w$.  For an NFA $\myautomaton$, we write $q \myto{a}{\myautomaton} q'$ for $q' \in \delta(q,a)$ and $q \mytto{w}{\myautomaton} q'$ for $q' \in \hat{\delta}(q,w)$.  We may omit $\myautomaton$'s when clear.

Let $\nfaalph = \Sigma \mydisjointu \syuugou[\lbrack_i,\rbrack_i,i]{1 \leq i \leq \kbound}$.  It is easy to see that, for any rewb $\alpha$, there is an NFA $\myautomaton_{\alpha} = (Q,\nfaalph,\delta,q_0,F)$ equivalent to $\reflang(\alpha)$ satisfying the following conditions: (a) every state in $\myautomaton_{\alpha}$ can reach some final state and can be reached from $q_0$, and (b) if either $q \myto{\lbrack_i}{} q'$ or $q \myto{\rbrack_i}{} q'$, then there is no other transition from $q$ to $q'$.
%
Let $n = \zettaiti{Q}$ and $Q = \syuugou{q_0,\dots,q_{n-1}}$.
\begin{lemma} \label{lem:om}
    Let $v, v_1, v_2 \in \nfaalph^\ast$.  Fix $q_j \in Q$. For any $v_1$ and $v_2$ such that $q_0 \mytto{v_1}{} q_j$ and $q_0 \mytto{v_2}{} q_j$, $\myopen(v_1) = \myopen(v_2)$. We define $M_j$ to be the unique set $\myopen(v)$ for some $q_0 \mytto{v}{} q_j$.
\end{lemma}

To prove the lemma, we first establish several properties of $\myopen$.

\begin{proposition} \label{prop:refworddyck} Let $v \in \reflang(\alpha)$ and $i \in \mynat$. If $v$ decomposes as $x_0 b_1 x_1 \cdots b_m x_m$ where $\lbrack_i, \rbrack_i \notin x_r$ and $b_r \in \{ \mathord{\lbrack_i}, \mathord{\rbrack_i} \}$, then $m$ is even, odd-numbered $b$ is $\lbrack_i$, and even-numbered $b$ is $\rbrack_i$.
\end{proposition}
\begin{proof} Immediate by induction on $\alpha$.
\end{proof}
\begin{corollary} \label{cor:prefixdyck} For any prefix $p$ of a string in $\reflang(\alpha)$, if $p$ decomposes as $x_0 b_1 x_1 \cdots b_m x_m$ where $\lbrack_i, \rbrack_i \notin x_r$ and $b_r \in \{ \mathord{\lbrack_i}, \mathord{\rbrack_i} \}$, odd-numbered $b$ is $\lbrack_i$ and even-numbered $b$ is $\rbrack_i$.
\end{corollary}
\begin{corollary} \label{cor:finalopisempty} For any $v \in \reflang(\alpha)$, $\myopen(v) = \emptyset$.
\end{corollary}
\begin{lemma} \label{lem:opassert} Let $p$ denote a nonempty prefix of some string in $\reflang(\alpha)$. If the last character of $p$ is $\lbrack_i$, then $i \notin \myopen(p^{\circ})$ and if it is $\rbrack_i$, then $i \in \myopen(p^{\circ})$, where $p^{\circ}$ denotes $p$ but with the last character erased.
\end{lemma}

\begin{proof} By Corollary~\ref{cor:prefixdyck} and induction on the number of times $\lbrack_i$ and $\rbrack_i$ occur in $p$.
\end{proof}

\begin{proof}[Proof of Lemma~\ref{lem:om}]
    Thanks to the NFA condition (a), the proof can be done by induction on the number of transitions from $q_{j}$ to some final state.  The base case is by Corollary~\ref{cor:finalopisempty}.  For the inductive case, assume that $q'$ can reach some final state such that the claim holds for $q'$. For all predecessor state $q$ of $q'$ (i.e., $q' \in \delta(q,c)$ for some $c \in \nfaalph$), we show that the claim also holds for $q$.  Take any strings $v_1, v_2$ such that $q_0 \mytto{v_1}{} q$ and $q_0 \mytto{v_2}{} q$. We have to prove $\myopen(v_1) = \myopen(v_2)$.  Suppose that $v_1$ extends to $v_1 c_1$ and $v_2$ to $v_2 c_2$ by moves from $q$ to $q'$.  Notice that the induction hypothesis assures $\myopen(v_1 c_1) = \myopen(v_2 c_2)$.  If $c_1 \in \Sigma \mydisjointu \mynat$, by the NFA condition (b), $c_2$ must not be a bracket and the claim holds. Otherwise, if $c_1 = \mathord{\lbrack_i}$, so is $c_2$ by the NFA condition (b), therefore $\myopen(v_1 \lbrack_i) = \myopen(v_2 \lbrack_i)$, namely $\myopen(v_1) \cup \syuugou{i} = \myopen(v_2) \cup \syuugou{i}$ holds. Because by Lemma~\ref{lem:opassert}, neither $\myopen(v_1)$ nor $\myopen(v_2)$ has $i$, $\myopen(v_1) = \myopen(v_2)$ holds as required. The case of $c_1 = \mathord{\rbrack_i}$ is similar.
\end{proof}

We now present the construction of a unary-PMCFG that proves the main theorem (Theorem~\ref{thm:rewbpmcfg}).  Let $\alpha$ be a rewb using number characters each of whose value is bounded by $\kbound$.
Let $G_\alpha$ be the $(\kbound+1)$-unary-PMCFG
\[
    (\syuugou{S, A_0, \dots, A_{n-1}},\Sigma,\mathcal{F},P,S)
\]
defined as follows:
\begin{itemize}
    \item $\mathcal{F} = \syuugou[\symo{},\symi_{a}^{M},\symp_{k}^{M},\symr_{k},\symid]{a \in \Sigma, 1 \leq k \leq \kbound, M \subseteq \syuugou{1,\dots,\kbound}}$, where
        \begin{itemize}
            \item $\symo = \lambda(x,y_1,\dots,y_\kbound).x$, $\symid = \lambda(x,y_1,\dots,y_\kbound).(x,y_1,\dots,y_\kbound)$ (the identity function),
            \item $\symi_{a}^M = \lambda(x,y_1,\dots,y_\kbound).(xa,z_1,\dots,z_\kbound)$ where $z_j$ is $y_j a$ for $j \in M$ and $y_j$ for $j \notin M$,
            \item $\symp_{k}^M = \lambda(x,y_1,\dots,y_\kbound).(xy_k,z_1,\dots,z_\kbound)$ where $z_j$ is $y_j y_k$ for $j \in M$ and $y_j$ for $j \notin M$,
            \item $\symr_{k}^M = \lambda(x,y_1,\dots,y_\kbound).(x,z_1,\dots,z_\kbound)$ where $z_k = \varepsilon$ and $z_j = y_j$ for $j \neq k$.
        \end{itemize}
        We use the following abbreviations: $\varphi_{a}^{M} = \symi_{a}^{M}$, $\varphi_{k}^{M} = \symp_{k}^{M}$, $\varphi_{\lbrack_k}^{M} = \symr_{k}$ and $\varphi_{\rbrack_k}^{M} = \symid$.
    \item $P = \syuugou[S \to \symo A_f]{q_f \in F}$ $\mydisjointu$ $\syuugou[A_j \to \varphi_c^{M_j} A_i]{q_i \myto{c}{\myautomaton_\alpha} q_j}$ $\mydisjointu$ $\syuugou{A_0 \to (\varepsilon, \varepsilon, \dots, \varepsilon)}$.
\end{itemize}

\begin{lemma} \label{lem:base}
    Let $v = v_1 \cdots v_{\zettaiti{v}} \in \reflang(\alpha)$ where $v_1, \dots, v_{\zettaiti{v}} \in \nfaalph$. Then, for every $i \geq 0$,
    \[
        \varphi_{v_i}^{\myopen(v_1\cdots v_i)} \cdots \varphi_{v_1}^{\myopen(v_1)} (\varepsilon, \varepsilon, \dots, \varepsilon) = (\deref(v_1 \cdots v_i), \mymem_1(v_1 \cdots v_i), \dots, \mymem_k(v_1 \cdots v_i)).
    \]
    In particular, $\symo\cdot\varphi_{v_{\zettaiti{v}}}^{\myopen(v_1\cdots v_{\zettaiti{v}})} \cdots \varphi_{v_1}^{\myopen(v_1)} (\varepsilon, \varepsilon, \dots, \varepsilon) = \deref(v)$.
\end{lemma}

\begin{proof}
By induction on $i$ using Proposition~\ref{prop:memind}.
\end{proof}

    \begin{proof}[Proof of Theorem~\ref{thm:rewbpmcfg}]
    It suffices to prove $L(\alpha) = L(G_\alpha)$.
    For proving $L(\alpha) \subseteq L(G_\alpha)$, take any $w \in L(\alpha)$ and fix $v = v_1 \cdots v_{\zettaiti{v}} \in \reflang(\alpha)$ ($v_1, \dots, v_{\zettaiti{v}} \in \nfaalph$) such that $\deref(v) = w$. Let $m = \zettaiti{v}$.  Fix a transition sequence $q_{\myperm_0} = q_0 \myto{v_1}{} q_{\myperm_1} \cdots \myto{v_m}{} q_{\myperm_m} \in F$. For each $i = 1, \dots, m$, we have $q_{\myperm_{i-1}} \myto{v_i}{} q_{\myperm_i}$, so $G_{\alpha}$ has the rule $A_{\myperm_i} \to \varphi_{v_i}^{M_{\myperm_i}} A_{\myperm_{i-1}}$.  Combine those as
    \[
        S \to \symo A_{\myperm_m} \to \symo \cdot \varphi_{v_m}^{M_{\myperm_m}} A_{\myperm_{m-1}} \to \cdots \to \symo \cdot \varphi_{v_m}^{M_{\myperm_m}} \cdots \varphi_{v_1}^{M_{\myperm_1}} A_{\myperm_0} \to \symo \cdot \varphi_{v_m}^{M_{\myperm_m}} \cdots \varphi_{v_1}^{M_{\myperm_1}} (\varepsilon,\varepsilon,\dots,\varepsilon).
    \]
    By Lemma~\ref{lem:om} and $q_0 \mytto{v_1\cdots v_i}{} q_{\myperm_{i}}$ for $i = 1, \dots, m$, the rightmost derived expression equals 
    \[
        \symo \cdot \varphi_{v_m}^{\myopen(v_1\cdots v_m)} \cdots \varphi_{v_1}^{\myopen(v_1)} (\varepsilon,\varepsilon,\dots,\varepsilon).
    \]
    By Lemma~\ref{lem:base}, it also equals $\deref(v) = w$. Hence, $G_{\alpha}$ derives $w$.
    Conversely, take any $w \in L(G_\alpha)$ and fix a derivation
    \[
        S \to \symo A_{\myperm_m} \to \symo \cdot f_{m}^{M_{\myperm_m}} A_{\myperm_{m-1}} \to \cdots \to \symo \cdot f_{m}^{M_{\myperm_m}} \cdots f_{1}^{M_{\myperm_1}} A_{\myperm_0} \to \symo \cdot f_{m}^{M_{\myperm_m}} \cdots f_{1}^{M_{\myperm_1}} (\varepsilon,\varepsilon,\dots,\varepsilon) = w.
    \]
    For each $i = 1, \dots, m$, $G_{\alpha}$ has the rule $A_{\myperm_i} \to f_{i}^{M_{\myperm_i}} A_{\myperm_{i-1}}$. From the construction of $G_{\alpha}$, there exists $v_i$ such that $\varphi_{v_i}^{M_{\myperm_i}} = f_i^{M_{\myperm_i}}$, $q_{\myperm_{i-1}} \myto{v_i}{} q_{\myperm_i}$, $q_{\myperm_0} = q_0$ and $q_{\myperm_m} \in F$. In particular, we have $v = v_1 \cdots v_m \in \reflang(\alpha)$. Because $q_0 \mytto{v_1 \cdots v_i}{} q_{\myperm_i}$ and Lemma~\ref{lem:om}, $M_{\myperm_i}$ coincides with $\myopen(v_1\cdots v_i)$. Therefore, $w = \symo\cdot\varphi_{v_m}^{\myopen(v_1\cdots v_m)} \cdots \varphi_{v_1}^{\myopen(v_1)} (\varepsilon, \varepsilon, \dots, \varepsilon)$ and the rightmost expression equals $\deref(v)$ by Lemma~\ref{lem:base}. Therefore, we have $w \in L(\alpha)$.
\end{proof}

We have shown that the language class of rewbs falls inside the class of unary-PMCFLs.  However, we also show a simple rewb with only one capturing group and no captured references that describes a non-MCFL. We need the pumping lemma for $m$-MCFLs established by Seki et al.~\cite{seki1991multiple}:
\begin{theorem}[{\cite[Lemma~3.2]{seki1991multiple}}]
\label{thm:pumpMCFL}
    For any infinite $m$-MCFL $L$, there exist $z_0 \in L$ and strings 
    $u_1, \dots, u_{2m+1}$, $v_1, \dots, v_{2m}$ that satisfy the following conditions:
    $(\mathrm{a})$ $z_0 = u_1 u_2 \cdots u_{2m} u_{2m+1}$, $(\mathrm{b})$ at least one of $v_1, \dots, v_{2m}$ is nonempty, and $(\mathrm{c})$ $z_i = u_1 v_1^i u_2 v_2^i \cdots u_{2m} v_{2m}^i u_{2m+1} \in L$ for all $i \geq 1$.
\end{theorem}
Note that this pumping lemma is somewhat weaker than ordinary ones. It just assures existence of such a string $z_0$ and does not say ``for all but finitely many strings $z_0$.''
\begin{theorem}
\label{thm:notmcfg}
    $L((_1 a^\ast )_1\,(c \bs 1)^\ast) = \syuugou[w(cw)^k]{k \geq 0, w \in \syuugou{a}^\ast}$ is not an MCFL.
\end{theorem}
\begin{proof} 
    Let $\alpha = (_1 a^\ast )_1\,(c \bs 1)^\ast$.
    Assume to the contrary that $L(\alpha)$ is an $m$-MCFL for some $m$. Let $r_m$ denote the regular expression $a^\ast c a^\ast c \cdots c a^\ast$, where $c$ occurs $2m$ times. Since $m$-MCFL is closed under intersections with regular languages~\cite{seki1991multiple}, $L \defeq L(\alpha) \cap L(r_m) = \syuugou[w (cw)^{2m}]{w \in \syuugou{a}^\ast}$ is also an $m$-MCFL. In what follows, we adopt the idea from the proof of Lemma~3.3 in \cite{seki1991multiple}. 
    Fix $z_0 = w(cw)^{2m} \in L$ and strings $u_1, v_1, \dots, u_{2m}, v_{2m}, u_{2m+1}$ that satisfy the conditions of Theorem~\ref{thm:pumpMCFL}.  We can write $z_1 = x(cx)^{2m}$ by condition (c).  Condition (b) implies that some $v_{i}$ is nonempty, hence $\zettaiti{x} > \zettaiti{w}$ must hold.
    On the other hand, observe that none of $v_1, \dots, v_{2m}$ has $c$ because if so, $z_1$ would have more than $2m$ $c$'s and contradict $z_1 \in L$. 
    By the pigeonhole principle, at least one $x$ of $2m+1$ $x$'s in $z_1 = x(cx)^{2m}$ has none of $2m$ possible $v_1, \dots, v_{2m}$, thus $x = w$ also must hold.
    This is a contradiction.
\end{proof}

As a corollary, we also obtain the following from the facts that the language class of rewbs is not a subclass of stack languages (SL) and every rewb without captured references describes a nonerasing stack language (NESL)\cite{mfcs2023,ic2025}:
\begin{corollary}
    \label{cor:pmcfgsaandmcfgnesa}
    \begin{enumerate}[(a)]
        \item The class of SLs does not contain that of unary-PMCFLs.
        \item The class of MCFLs does not contain that of NESLs.
    \end{enumerate}
\end{corollary}

    In what follows, we show that a CFL known to lie outside EDT0L is in NL.  As remarked in the introduction, this shows that $\text{EDT0L} \subsetneq \text{NL} \cap \text{IL}$.  Let $B_1$ be the CFL over the alphabet $\{\mathord{[},\mathord{]}\}$ generated by the context-free grammar with the single nonterminal $S$ and production rules $S \to [SS] \mid \varepsilon$.
Ehrenfeucht and Rozenberg showed that $B_1$ is not an EDT0L language~\cite{ehrenfeucht1977some}.  

\begin{proposition} \label{prop:b1nl}
    $B_1$ is decidable in nondeterministic logspace.
\end{proposition}
\begin{proof}
    We give a nondeterministic logspace Turing machine that decides the complement of $B_1$.  The proposition then follows from the Immerman-Szelepcs\'{e}nyi theorem~\cite{DBLP:journals/siamcomp/Immerman88,DBLP:journals/acta/Szelepcsenyi88}.  Let $D$ be the Dyck language over the alphabet $\{\mathord{[},\mathord{]}\}$ generated by the context-free grammar with the single nonterminal $S$ and production rules $S \to SS \mid [S] \mid \varepsilon$.  Note that $B_1 \subseteq [D] = \{ [ u ] \mid u \in D \} \cup \{ \varepsilon \}$.  
    It is easy to see that every $u \in D$ has a unique decomposition $u = [u_1] [u_2] \cdots [u_l]$ for some $l \ge 0$ and $u_1, u_2, \dots, u_l \in D$.  
    Given this decomposition, we observe that $[u] \in B_1$ if and only if $l \le 2$ and $[u_1], \dots, [u_l] \in B_1$.
    Based on this observation, our Turing machine works as follows:
    \begin{enumerate}
        \item If the input $w$ is the empty string $\varepsilon$, we reject.  Otherwise, we check whether $w \in [D]$ and if $w \notin [D]$, we accept.  This can be done in deterministic logspace by using a counter, incrementing it by one when reading an opening bracket and decrementing it by one when reading a closing bracket.    
        \item We may assume that $w = [u]$ for some $u \in D$. If $u = \varepsilon$, we reject.  Otherwise, we uniquely decompose $u$ into $[u_1] [u_2] \cdots [u_l]$ for some $l \ge 1$ and $u_1, u_2, \dots, u_l \in D$ using the counter method from Step~1.  If $l > 2$, we accept.  Otherwise, we have $l = 1, 2$.  If $l = 1$ and $u = [u_1]$, we repeat this step with $u$ in place of $w$.  If $l = 2$ and $u = [u_1] [u_2]$, we guess one of $[u_1]$ and $[u_2]$ that is not in $B_1$ and repeat this step with it in place of $w$.
    \end{enumerate}
    Clearly, this decides the complement of $B_1$ and runs in nondeterministic logspace.
\end{proof}

\section{Closed-star rewbs and unary-MCFLs}
\label{sec:rewbscmcfg}

As we have shown in the previous section, the language class of rewbs falls into the class of unary-PMCFLs, but not the class of MCFLs.
Therefore, the parallelism of PMCFG is essential for the inclusion.  
To see why, let us recall the construction of the unary-PMCFG $G_\alpha$ equivalent to the rewb $\alpha = (_1 a^\ast )_1\,(c \bs 1)^\ast$ from Example~\ref{ex:pmcfgproof}.
Note that only $\symp_{1}^{\syuugou{}} = \lambda (x,y_1). (x y_1, y_1)$ is parallel and we cannot deprive its parallelism because $L(\alpha)$ is not an MCFL as shown in Theorem~\ref{thm:notmcfg}.
Roughly, this is due to the fact that the subexpression $(c \bs 1)^\ast$ may reference memory cell $1$ defined outside the star from the inside of the star.
When processing number character $1$, $G_{\alpha}$ needs to not only append the content of memory cell $1$ to the generating string but also keep the content of memory cell $1$ for the future, because $G_{\alpha}$ cannot foresee how many times memory cell $1$ will be referenced in the future derivation.  
Backed by these observations, in this section, we propose a novel syntactic condition that we call \emph{closed-star} which asserts that no subexpression inside a star can backreference strings captured outside the star.
We show that the closed-star condition gives an upper bound on the number of times each content of memory cells will be referenced, and prove that the language class of closed-star rewbs falls inside the class of unary-MCFLs.

Formally, we call a rewb \emph{closed} if every reference $\bs i$ in the rewb has the corresponding capturing group $(_i\,)_i$.
For instance, $(_1 a^\ast )_1 \bs 1$ is closed but $(_1 a^\ast )_1 \bs 2$ and $\bs 1$ are not because neither $\bs 2$ in the former nor $\bs 1$ in the latter has the corresponding capturing group.
Formally, we cite the following definition of closedness (originally, NUR condition) of rewbs~\cite{chida2022lookaheads}:\footnote{\cite{chida2022lookaheads} used the name ``no unassigned references (NUR)'' to refer to the property following the prior work~\cite{campeanu2003formal,carle2009extended} that also studied this property.  We prefer the name ``closed'' because, while those prior work globally enforced the property on the entire given rewb, we assert the property locally on each of certain subexpressions of the rewb (i.e., those inside stars) rather than globally on the entire rewb.}
\begin{definition2} 
    \label{def:closed}
    Let $\alpha$ be a rewb.
    \begin{enumerate}[(a)]
        \item Define the mapping $\mycapt$ as follows: $\mycapt(a) = \mycapt(\varepsilon) = \mycapt(\bs i) = \emptyset$, $\mycapt(\alpha^\ast) = \mycapt(\alpha)$, $\mycapt(\alpha_1 \alpha_2) = \mycapt(\alpha_1) \cup \mycapt(\alpha_2)$, $\mycapt(\alpha_1 + \alpha_2) = \mycapt(\alpha_1) \cap \mycapt(\alpha_2)$ and $\mycapt(\,(_i \alpha )_i\,) = \mycapt(\alpha) \cup \syuugou{i}$, where $a \in \Sigma$ and $i \in \mynat$. Roughly, $\mycapt(\alpha)$ is the set of number characters that are guaranteed to be bound in the continuation.\footnote{By ``continuation'', we mean that capturing groups can be viewed as variable definitions in programming.  For example, when $\alpha = \alpha_1 \alpha_2$ where $\alpha_1 = (_1 a^\ast )_1$ and $\alpha_2 = \bs 1$, we can view $\alpha_2$ as the continuation of $\alpha_1$ and $\alpha_1$ can be seen as defining the content of variable $1$ for its continuation.}
        \item Let $S$ denote a subset of $\mynat$.  Define the relation $\mycl{S}{\alpha}$ as follows: $\mycl{S}{a}$, $\mycl{S}{\varepsilon}$ and
            \begin{gather*}
                \infer{\mycl{S}{\bs i}}{i \in S}
                \quad
                \infer{\mycl{S}{\alpha^\ast}}{\mycl{S}{\alpha}}
                \quad
                \infer{\mycl{S}{\alpha_1 \alpha_2}}{\mycl{S}{\alpha_1} & \mycl{S\cup \mycapt(\alpha_1)}{\alpha_2}}
                \quad
                \infer{\mycl{S}{\alpha_1 + \alpha_2}}{\mycl{S}{\alpha_1} & \mycl{S}{\alpha_2}}
                \quad
                \infer{\mycl{S}{(_i \alpha )_i}}{\mycl{S}{\alpha}}
            \end{gather*}
            where $a \in \Sigma$ and $i \in \mynat$.  Roughly, $\mycl{S}{\alpha}$ if and only if $S$ contains every number character referenced externally by $\alpha$.\footnote{We say that a rewb subexpression references a number character $i$ externally if the corresponding capturing group $(_i\,)_i$ to be backreferenced is not in the subexpression.  Again, when $\alpha = \alpha_1 \alpha_2$ where $\alpha_1 = (_1 a^\ast )_1$ and $\alpha_2 = \bs 1$ for example, $\alpha_2$ references $1$ externally because the corresponding capturing group $(_1 a^\ast )_1$ is not inside $\alpha_2$ itself but is inside $\alpha_1$.}
            We call $\alpha$ \emph{closed} if $\mycl{\emptyset}{\alpha}$.
    \end{enumerate}
\end{definition2}
Here is our definition of closed-star rewbs.
\begin{definition2} \label{def:sc}
    A rewb is \emph{closed-star} if for each its starred subexpression $\beta^\ast$, $\beta$ is closed.
\end{definition2}

For instance, $(_1 a^\ast )_1 c \bs 1$ and $((_1 a^\ast )_1 c \bs 1)^\ast$ are closed-star but $(_1 a^\ast )_1 (c \bs 1)^\ast$ is not even though the rewb itself is (i.e., globally) closed.  Conversely, $\bs 1 + (_1 a^\ast )_1 c \bs 1$ is closed-star even though the rewb itself is not closed.
It is easy to see that the language class of closed-star rewbs contains the class of regular languages but is not contained in that of CFLs as it can describe, for example, the copy language (cf.~Example~\ref{ex:copylang}).
Our goal is to prove the following:

\begin{restatable}{theorem}{thmrewbscmcfg} \label{thm:rewbscmcfg}
    Every closed-star rewb describes a unary-MCFL.
\end{restatable}

The proof idea is as follows.  Let $\alpha$ be a closed-star rewb using number characters each of whose value is bounded by $\kbound$.
Recall that in the previous section, we constructed an equivalent $(\kbound+1)$-unary-PMCFG $G_\alpha$ using a tuple $(x,y_1,\dots,y_\kbound)$ in a derivation.  
This time, we construct an equivalent $(\kbound\cdot\rbound + 1)$-unary-MCFG $G'_\alpha$ using a tuple $(x,y_{1,1},\dots,y_{1,\rbound},\dots,y_{\kbound,1},\dots,y_{\kbound,\rbound})$ for some $\rbound \in \mynat$.  We let $G'_{\alpha}$ have $\rbound$ copies for each component $y_k$ $(1 \leq k \leq \kbound)$ of $G_\alpha$ to handle multiple references to the same string without parallelism.  The hard part of the proof is showing that such a finite $\rbound$ exists.

Fix $v = v_1 \cdots v_{\zettaiti{v}} \in \reflang(\alpha)$ where $v_1, \dots, v_{\zettaiti{v}} \in \nfaalph = \Sigma \mydisjointu \syuugou[\lbrack_i,\rbrack_i,i]{1 \leq i \leq \kbound}$.  
We introduce two tuples $t_i = (t_{i,1},\dots,t_{i,\kbound})$ and $s_i = (s_{i,1},\dots,s_{i,\kbound})$ for each position $i = 0, 1, \dots, \zettaiti{v}$ for calculating $\rbound$.  Intuitively, for each $1 \leq k \leq \kbound$, $t_{i,k}$ denotes the number of times the content of memory cell $k$ after dereferencing $v_1 \cdots v_i$ will be directly referenced when dereferencing the remaining string $v_{i+1} \cdots v_{\zettaiti{v}}$.
Similarly, $s_{i,k}$ denotes the ``dependency depth'' of memory cell $k$ at position $i$, that is, roughly, a degree of depth to which the content of cell $k$ after dereferencing $v_1 \cdots v_i$ will be passed to the other memory cells or the generating string.
For example, let $v = \lbrack_1 aa \rbrack_1 \lbrack_2\, 1\, 1\,\rbrack_2 \lbrack_1\,2\,\rbrack_1\,1$.
Then, we calculate $t_{5,1} = 2$ and $s_{5,1} = 3$ (see below for the formal definition).  This means that at position $5$ (i.e.,~between $\lbrack_2$ and $1$), the current string $aa$ in cell $1$ will be directly referenced twice and passed three times (in the order: $\to$ cell~$2$ (and generating string) $\to$ cell~$1$ (and generating string) $\to$ generating string).
We formally define the tuples $t_i$ and $s_i$ inductively on the positions of $v$.

\begin{definition2} \label{def:t}
    Let $v = v_1 \cdots v_{\zettaiti{v}} \in \nfaalph^\ast$ where $v_1, \dots, v_{\zettaiti{v}} \in \nfaalph$.  Define $t_i^v$ (or $t_i$ when $v$ is clear) for $i = 0, 1, \dots, \zettaiti{v}$ inductively as follows: (a) $t_{\zettaiti{v}} = (0, \dots, 0)$ ($\kbound$ times), (b) if $v_i$ is a character in $\Sigma$ or $\rbrack_j$, then $t_{i-1} = t_i$, (c) if $v_i$ is $\lbrack_j$, then $t_{i-1,k} = t_{i,k}$ for $k \neq j$ and $t_{i-1,j} = 0$, and (d) if $v_i$ is a number character $j$, then $t_{i-1,k} = t_{i,k}$ for $k \neq j$ and $t_{i-1,j} = t_{i,j} + 1$.
    Let $\tau$ and $\tau'$ be tuples of dimension $\kbound$.  We write $\tau \tarrow{c} \tau'$ when $\tau$ is determined by $c \in \nfaalph$ and $\tau'$ by the inductive step (i.e., as how $t_{i-1}$ is determined by $v_i$ and $t_i$).
\end{definition2}

\begin{definition2} \label{def:s}
    Let $v = v_1 \cdots v_{\zettaiti{v}} \in \nfaalph^\ast$ where $v_1, \dots, v_{\zettaiti{v}} \in \nfaalph$.  Define $s_i^v$ (or $s_i$ when $v$ is clear) for $i = 0, 1, \dots, \zettaiti{v}$ inductively as follows: (a) $s_{\zettaiti{v}} = (0, \dots, 0)$ ($\kbound$ times), (b) if $v_i$ is a character in $\Sigma$ or $\rbrack_j$, then $s_{i-1} = s_i$, (c) if $v_i$ is $\lbrack_j$, then $s_{i-1,k} = s_{i,k}$ for $k \neq j$ and $s_{i-1,j} = 0$, and (d) if $v_i$ is a number character $j$, then $s_{i-1,k} = s_{i,k}$ for $k \neq j$ and $s_{i-1,j}$ is $\max \{ 1, s_{i,j} \}\cup \syuugou[s_{i,l}+1]{l \in \myopen(v_1\cdots v_i)}$.
    Let $\sigma$ and $\sigma'$ be tuples of dimension $\kbound$.  We write $\sigma \sarrow{c} \sigma'$ when $\sigma$ is determined by $c \in \nfaalph$ and $\sigma'$ by the inductive step as in Definition~\ref{def:t}.
\end{definition2}

As we will show in Proposition~\ref{prop:tconst} (resp.~Proposition~\ref{prop:sconst}), the rewb $\alpha$ being closed-star ensures the existence of upper bounds $\tconst$ (resp.~$\sconst$) on all components $t_{i,k}$ (resp.~$s_{i,k}$) no matter what $v \in \reflang(\alpha)$ is considered.
Using $\tconst$ and $\sconst$, we define $\mydepth_l$ as $\mydepth_0 = 0$ and $\mydepth_l = \tconst (1 + \kbound\cdot\mydepth_{l-1} )$ inductively, and define $\rbound = \mydepth_{\sconst}$.  Note that $\rbound$ is finite.

For each $1 \leq k \leq \kbound$, we shall design the subcomponents $(y_k) = (y_{k,1},\dots,y_{k,\rbound})$ as shown in the diagram below.\footnote{We might as well write $(y_k)$ as $y_k$, but we shall use the former notation to emphasize that it is a tuple.}
This design is intended for $G'_\alpha$ to keep the contents of the components from $y_{k,1}$ to $y_{k,\psi_i(k,\kbound)}$ (where $\psi_i(k,\kbound) = (t_{i,k}/\tconst) \mydepth_{s_{i,k}}$) to be the copies of $\mu_k(v_1\cdots v_i)$ immediately before applying the function corresponding to $v_{i+1}$ in the functional dereferencing phase.\footnote{For the term ``functional dereferencing phase,'' see its description given in the paragraph immediately below Theorem~\ref{thm:rewbpmcfg} and Example~\ref{ex:pmcfgproof} in the previous section.}
The case $v_{i+1} \in \mynat$ is the most difficult case for preserving this intended structure.  Let us look at the diagram.  In this case, the function $\symp_{k}^{\myopen(v_1\cdots v_{i+1}),t_i,s_i}$ corresponding to $v_{i+1}=k$ (cf.~the definition of $G'_\alpha$ before Lemma~\ref{lem:mcfgbase}) (1) appends the first component of the $t_{i,k}$-th block of $(y_k)$ to the first component of the entire tuple which is maintaining the generating string, (2) for each open memory cell $j \in \myopen(v_1 \cdots v_{i+1})$, appends each component $y_{k,r}$ $(\psi_i(k,j-1) \leq r \leq \psi_i(k,j))$ of the subblock to the first $\mydepth_{s_{i,k}-1}$ components of $(y_j)$, and (3) empties all components of the $t_{i,k}$-th block of $(y_k)$ lest non-parallelism would be violated.

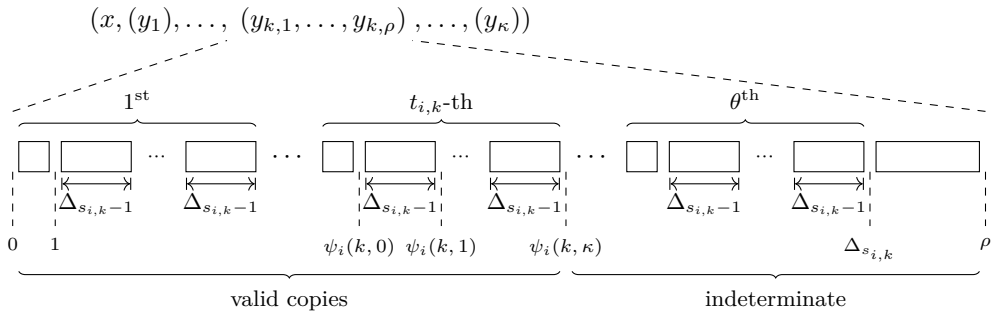
\begin{figure}[htb]
    \centering
    \begin{tikzpicture}[x=4mm,y=4mm]
        \tikzset{
            position label/.style={
               above = 0.4pt,
               text height = 1ex,
               text depth = 0.6ex
            },
            position label under/.style={
               below = 3pt,
               text height = 1.5ex,
               text depth = 1ex
            },
            brace/.style={
               decoration={brace},
               decorate
            },
            underbrace/.style={
               decoration={brace, mirror},
               decorate
            }
        }
        \node[below=0.35mm] (x1) at (0.2,-1.8) {\scriptsize $0$};
        \draw[dashed] (0.2,0) -- (0.2,-1.8);
        \draw (0.4,0) rectangle (1.4,1);
        \node[below=0.35mm] (x2) at (1.5999999999999999,-1.8) {\scriptsize $1$};
        \draw[dashed] (1.5999999999999999,0) -- (1.5999999999999999,-1.8);
        \draw (1.7999999999999998,0) rectangle (4.1,1);
        \draw[|<->|] ($ (1.7999999999999998,0) + (0,-0.4) $) -- node[below] {\footnotesize $\mydepth_{s_{i,k}-1}$} ($ (4.1,0) + (0,-0.4) $);
        \node (b1) at (5.0,0.5) {$\scriptstyle\cdots$};
        \draw (5.9,0) rectangle (8.2,1);
        \draw[|<->|] ($ (5.9,0) + (0,-0.4) $) -- node[below] {\footnotesize $\mydepth_{s_{i,k}-1}$} ($ (8.2,0) + (0,-0.4) $);
        \node (c1) at (9.299999999999999,0.5) {$\cdots$};
        \draw (10.4,0) rectangle (11.4,1);
        \node[below=0.35mm] (x7) at (11.6,-1.8) {\scriptsize $\psi_i(k,0)$};
        \draw[dashed] (11.6,0) -- (11.6,-1.8);
        \draw (11.8,0) rectangle (14.100000000000001,1);
        \draw[|<->|] ($ (11.8,0) + (0,-0.4) $) -- node[below] {\footnotesize $\mydepth_{s_{i,k}-1}$} ($ (14.100000000000001,0) + (0,-0.4) $);
        \node[below=0.35mm] (x8) at (14.3,-1.8) {\scriptsize $\psi_i(k,1)$};
        \draw[dashed] (14.3,0) -- (14.3,-1.8);
        \node (b2) at (15.000000000000002,0.5) {$\scriptstyle\cdots$};
        \draw (15.900000000000002,0) rectangle (18.200000000000003,1);
        \draw[|<->|] ($ (15.900000000000002,0) + (0,-0.4) $) -- node[below] {\footnotesize $\mydepth_{s_{i,k}-1}$} ($ (18.200000000000003,0) + (0,-0.4) $);
        \node[below=0.35mm] (x10) at (18.400000000000002,-1.8) {\scriptsize $\psi_i(k,\kbound)$};
        \draw[dashed] (18.400000000000002,0) -- (18.400000000000002,-1.8);
        \node (c2) at (19.3,0.5) {$\cdots$};
        \draw (20.4,0) rectangle (21.4,1);
        \draw (21.799999999999997,0) rectangle (24.099999999999998,1);
        \draw[|<->|] ($ (21.799999999999997,0) + (0,-0.4) $) -- node[below] {\footnotesize $\mydepth_{s_{i,k}-1}$} ($ (24.099999999999998,0) + (0,-0.4) $);
        \node (b3) at (24.999999999999996,0.5) {$\scriptstyle\cdots$};
        \draw (25.899999999999995,0) rectangle (28.199999999999996,1);
        \draw[|<->|] ($ (25.899999999999995,0) + (0,-0.4) $) -- node[below] {\footnotesize $\mydepth_{s_{i,k}-1}$} ($ (28.199999999999996,0) + (0,-0.4) $);
        \node[below=0.35mm] (x15) at (28.399999999999995,-1.8) {\scriptsize $\mydepth_{s_{i,k}}$};
        \draw[dashed] (28.399999999999995,0) -- (28.399999999999995,-1.8);
        \draw (28.599999999999994,0) rectangle (32.0,1);
        \node[below=0.35mm] (x16) at (32.2,-1.8) {\scriptsize $\rbound$};
        \draw[dashed] (32.2,0) -- (32.2,-1.8);
        \draw [underbrace] (0.4,-3.3) -- node [position label under, pos = 0.5] {\footnotesize valid copies} (18.200000000000003,-3.3);
        \draw [underbrace] (18.6,-3.3) -- node [position label under, pos = 0.5] {\footnotesize indeterminate} (32,-3.3);
        \draw [brace] (0.4,1.5) -- node [position label, pos = 0.5] {\footnotesize 1\textsuperscript{st}} (8.2,1.5);
        \draw [brace] (10.4,1.5) -- node [position label, pos = 0.5] {\footnotesize $t_{i,k}$-th} (18.2,1.5);
        \draw [brace] (20.4,1.5) -- node [position label, pos = 0.5] {\footnotesize $\tconst$\textsuperscript{th}} (28.2,1.5);
        \node (x) at (10,5) {$(x, (y_1), \dots,\;(y_{k,1},\dots,y_{k,\rbound})\;, \dots, (y_\kbound))$};
        \draw[dashed] (0.2,2) -- (7.25,4.35);
        \draw[dashed] (32.2,2) -- (13.35,4.35);
    \end{tikzpicture}
\end{figure}
\begin{example} \label{ex:mcfgproof}
    We give an example of the construction of $(\kbound\cdot\rbound+1)$-unary-MCFG $G'_\alpha$.
    Let $\alpha = \big( (_1 a^\ast )_1 (_2 \bs 1 )_2 \bs 2 \bs 2\big)^\ast$ and $\kbound=2$. 
    We assure that $\tconst = \sconst = 2$ is sufficient for $\alpha$, and thus $\rbound = \mydepth_2 = 10$.
    We use the following NFA $\myautomaton_\alpha = (Q,\nfaalph,\delta,q_0,F)$ which recognizes $\reflang(\alpha)$:
\begin{figure}[htb]
        \centering
        \begin{tikzpicture}[
                shorten >=1pt,
                node distance=2.25cm,
                on grid,
                >={Stealth[round]},
                auto
            ]
            \tikzset{el/.style = {inner sep=2pt, align=left, sloped}}
            \tikzset{every node/.style={font=\small}}
            \tikzset{every loop/.style={min distance=6mm,looseness=8}}
            \tikzset{every state/.style={inner sep=.1mm, minimum size=5mm}}
            \tikzset{initial text=}
            \tikzset{every initial by arrow/.style={->}}
            \tikzset{bend angle=30}

            \node[state,initial,accepting] (0) at (0,-3) {$0$};
            \node[state] (1) at (1,-2.3) {$1$};
            \node[state] (2) at (2,-3) {$2$};
            \node[state] (3) at (3,-2.3) {$3$};
            \node[state] (4) at (4,-2.3) {$4$};
            \node[state] (5) at (5,-3) {$5$};
            \node[state] (6) at (6,-3) {$6$};

            \path[->] (0) edge node [yshift=-2pt] {\footnotesize
                            $\lbrack_1$\!
                        } (1);
            \path[->] (1) edge [in=10,out=170,loop above] node {\footnotesize
                            $a$
                        } ();
            \path[->] (1) edge node [yshift=-2pt] {\footnotesize
                            $\rbrack_1$
                        } (2);
            \path[->] (2) edge node [yshift=-2pt] {\footnotesize
                            $\lbrack_2$\!
                        } (3);
            \path[->] (3) edge node {\footnotesize
                            $1$
                        } (4);
            \path[->] (4) edge node [yshift=-2pt] {\footnotesize
                            $\rbrack_2$
                        } (5);
            \path[->] (5) edge node {\footnotesize
                            $2$
                        } (6);
            \draw[->, rounded corners=5] (6) |- node {} ++(-0.7,-0.7) -| (0);
            \node (label) at (6.25,-3.5) {\footnotesize $2$};
        \end{tikzpicture}
\end{figure}

    The nonterminals are the initial nonterminal $S$ and nonterminals $A_i^{(\tau_1,\tau_2),(\sigma_1,\sigma_2)}$ for all $0 \leq i < \zettaiti{Q} = 7$, $0 \leq \tau_1,\tau_2 \leq \tconst = 2$ and $0 \leq \sigma_1,\sigma_2 \leq \sconst = 2$.  In a derivation, the two superscripts of each nonterminal except $S$ would represent the values of $t$ and $s$ at the string position corresponding to its nonterminal.
    The functions, as in Example~\ref{ex:pmcfgproof}, are divided into the five kinds (i.e., $\symo$, $\symi$, $\symp$, $\symr$ and $\symid$) but with some annotated with additional information.
    The rules are also similar to those in Example~\ref{ex:pmcfgproof} but this time the grammar has to update the superscripts to respect the character on the corresponding transition edge in $\myautomaton_\alpha$.
    We introduce a notation for tuples used only here: $(x)_m$ stands for $(x,\dots,x)$ ($m$ times) and $(x)_m(y)_n$ stands for the tuple of which the first $m$ are $x$'s and the last $n$ are $y$'s.
    $G'_\alpha$ derives $a^4$ backed by $v = \lbrack_1 a \rbrack_1 \lbrack_2\,1\,\rbrack_2\,2\,2 \in \reflang(\alpha)$ as follows:
    $S \to 
     \symo\cdot A_{0}^{(0,0),(0,0)} \to
     \symo\cdot\symp_2^{\{\},(0,1),(0,1)} A_{6}^{(0,1),(0,1)} 
    $, where the superscripts of the nonterminals are updated and provided to $\symp_2$, and continues as
    $ \to \cdots \to
      \symo\cdot\symp_2^{\{\},(0,1),(0,1)}\cdot\symp_2^{\{\},(0,2),(0,1)}\cdot\symid\cdot\symp_1^{\{2\},(1,2),(2,1)}\cdot\symr_2\cdot\symid\cdot\symi_a^{\{1\}}\cdot\symr_1(\varepsilon,(\varepsilon)_{10},(\varepsilon)_{10})
      = \symo\cdot\symp_2^{\{\},(0,1),(0,1)}\cdot\symp_2^{\{\},(0,2),(0,1)}\cdot\symp_1^{\{2\},(1,2),(2,1)}(a,(a)_{10},(\varepsilon)_{10})
      = \symo\cdot\symp_2^{\{\},(0,1),(0,1)}\cdot\symp_2^{\{\},(0,2),(0,1)}(a^2,(\varepsilon)_{1}(a)_{2}(\varepsilon)_{2}(a)_{5},(a)_{2}(\varepsilon)_8)
      = \symo\cdot\symp_2^{\{\},(0,1),(0,1)}(a^3,\dots,(a)_{1}(\varepsilon)_9)
      = \symo(a^4,\dots,(\varepsilon)_{10})
      = a^4
    $.
    Remark that $\symp_1^{\{2\},(1,2),(2,1)}$ conveys one copy of the captured string $a$ from cell~$1$ to the generating string, and two copies from cell~$1$ to cell~$2$ for the two further use of cell~$2$.
\end{example}


In what follows, we shall formally present our construction of a unary-MCFG equivalent to $L(\alpha)$.  We first show the existence of upper bounds $\tconst$ and $\sconst$ for $t$ and $s$, respectively, as stated in the paragraph immediately below Definition~\ref{def:s}.
To this end, we show the following properties of $\mycapt$ and closed rewbs.

\begin{claim} \label{clm:capt}
    Let $\alpha$ be a rewb and $i \in \mynat$.  If $i \in \mycapt(\alpha)$, then every $v \in \reflang(\alpha)$ has an $i$-bracketed substring $\lbrack_i v' \rbrack_i$.
\end{claim}

\begin{proof}
    By straightforward induction on $\alpha$.
\end{proof}
\begin{proposition} \label{prop:stronglymatching}
    Let $\alpha$ be a closed rewb and $i \in \mynat$.  If $x i y \in \reflang(\alpha)$, then the prefix $x$ has an $i$-bracketed substring $\lbrack_i v \rbrack_i$.
\end{proposition}

\begin{proof}
    We prove that the following stronger statement by induction on the depth of the derivation tree for $\mycl{S}{\alpha}$: if $\mycl{S}{\alpha}$ holds, then $\forall v \in \reflang(\alpha).\;\forall i \notin S.\;\forall x,y.\; v = x i y \Longrightarrow x$ has an $i$-bracketed substring.  
    We consider the case $\alpha = \alpha_1 \alpha_2$ (the other cases are straightforward).   
    Take any $v = v_1 v_2 \in \reflang(\alpha)$ where $v_l \in \reflang(\alpha_l)$ for $l = 1, 2$ and any $i \notin S$.  Assume $i \in v$. If $i \in v_1$, then the claim holds by the induction hypothesis for $\alpha_1$.  Otherwise (i.e., $i \in v_2$), we have the following subcases: 
    for the case $i \notin \mycapt(\alpha_1)$, the claim holds by the induction hypothesis for $\alpha_2$, and for the other case $i \in \mycapt(\alpha_1)$, by Claim~\ref{clm:capt}, $v_1$ has an $i$-bracketed substring.
\end{proof}
We are now ready to prove the existence of $\tconst$ and $\sconst$.
\begin{proposition} \label{prop:tconst}
    Suppose that $\alpha$ is closed-star.
    There exists $\tconst$ such that for all $v \in \reflang(\alpha)$, $k \in \{1, \dots, \kbound\}$ and $i \in \{ 0, 1, \dots, \zettaiti{v} \}$, $t_{i,k} \leq \tconst$ holds.
\end{proposition}

\begin{proof}
    Let $t^{v,(p_1,\dots,p_\kbound)}$ denote $t^{v}$ but with the initial value $(p_1,\dots,p_\kbound)$ instead of $(0,\dots,0)$.  We prove by induction on $\alpha$ the following stronger claim: for all $p_1, \dots, p_\kbound$, there exists $\tconst$ such that for all $v \in \reflang(\alpha), k$ and $i$, the inequality $t^{v,(p_1,\dots,p_\kbound)}_{i,k} \leq \tconst$ holds.  Notice that letting all $p_1, \dots, p_\kbound$ be zeros coincides with the claim of the proposition.
    We consider the cases $\alpha = \alpha_1 \alpha_2$ and $\alpha = \alpha_1 ^\ast$ (the other cases are straightforward).
    
        Case $\alpha = \alpha_1 \alpha_2$:  Take $p_1, \dots, p_\kbound$.  Fix $\tconst_2$ by the induction hypothesis for $\alpha_2$ over $p_1,\dots,p_\kbound$ and $\tconst_1$ by the induction hypothesis for $\alpha_1$ over $\tconst_2, \dots, \tconst_2$ ($\kbound$ times).  Introduce $v \in \reflang(\alpha), k$ and $i$.  Decompose $v = v_1 v_2$ where $v_l \in \reflang(\alpha_l)$ for $l = 1, 2$.  
            Observe by easy induction on $i$ that $t^{v_1,(p'_1,\dots,p'_\kbound)}_{i,k} \leq t^{v_1,(\tconst_2,\dots,\tconst_2)}_{i,k}$ for $i = 0,1,\dots,\zettaiti{v_1}$ where $(p'_1,\dots,p'_\kbound) = t^{v_2,(p_1,\dots,p_\kbound)}_0$.  Therefore, we obtain $t^{v,(p_1,\dots,p_\kbound)}_{i,k} = t^{v_1,(p'_1,\dots,p'_\kbound)}_{i,k} \leq t^{v_1,(\tconst_2,\dots,\tconst_2)}_{i,k} \leq \tconst_1$ by the induction hypothesis for $\alpha_1$.  Take $\tconst = \max\syuugou{\tconst_1, \tconst_2}$.

        Case $\alpha = \alpha_1 ^\ast$:  Take $p_1, \dots, p_\kbound$.  Fix $\tconst_1^{\bar{p_1},\dots,\bar{p_\kbound}}$ by the induction hypothesis of $\alpha_1$ for $\bar{p_k} \in \syuugou{p_k,0}$ where $k \in \syuugou{ 1, \dots, \kbound}$.  Let $\tconst$ be the maximum of the union of $\{ \tconst_1^{\bar{p_1},\dots,\bar{p_\kbound}} \mid \bar{p_k} \in \{ p_k,0 \}, k \in \{ 1, \dots, \kbound \} \}$ and $\syuugou[p_k]{k \in \syuugou{1,\dots,\kbound}}$.
            Take $v = v_1 \cdots v_m \in \reflang(\alpha)$ where each $v_r \in \reflang(\alpha_1)$.  We prove by induction on $m$. The base case $m = 0$ is obvious.  For the inductive case, assume that the inequality of the claim holds for $v_2 \cdots v_m$.  For each number character $j$, if $j \in v_2 \cdots v_m$, then fix the smallest $r \geq 2$ such that $j \in v_r$.  Since $\alpha_1$ is closed, by Proposition~\ref{prop:stronglymatching}, there is a $j$-bracketed substring in $v_r$ that appears before the $j$.  Therefore, from Definition~\ref{def:t}, $t^{v,(p_1,\dots,p_\kbound)}_{\zettaiti{v_1},j}$ is zero.  Otherwise (i.e., $j \notin v_2 \cdots v_m$), $t^{v,(p_1,\dots,p_\kbound)}_{\zettaiti{v_1},j}$ is zero or $p_j$.
          Therefore, 
            there exists $\bar{p_k} \in \{ p_k, 0\}$ for each $k \in \{1,\dots,\kbound\}$ such that for all $l \in \{ 0,1,\dots, \zettaiti{v_1} \}$,
            $t^{v,(p_1,\dots,p_\kbound)}_{l,j} = t^{v_1,(\bar{p_1},\dots,\bar{p_\kbound})}_{l,j}$.
          Hence, $\tconst$ is sufficient for the claim.
\end{proof}

\begin{proposition} \label{prop:sconst}
    Suppose that $\alpha$ is closed-star.
    There exists $\sconst$ such that for all $v \in \reflang(\alpha)$, $k \in \{1, \dots, \kbound\}$ and $i \in \{ 0, 1, \dots, \zettaiti{v} \}$, $s_{i,k} \leq \sconst$ holds.
\end{proposition}

\begin{proof}
    As with Proposition~\ref{prop:tconst}, let $s^{v,(p_1,\dots,p_\kbound)}$ denote $s^{v}$ but with the initial value $(p_1,\dots,p_\kbound)$ instead of $(0,\dots,0)$.  We prove by induction on $\alpha$ the following stronger claim: for all $p_1, \dots, p_\kbound$, there exists $\sconst$ such that for all $v \in \reflang(\alpha), k$ and $i$, the inequality $s^{v,(p_1,\dots,p_\kbound)}_{i,k} \leq \sconst$ holds.
    Letting all $p_1, \dots, p_\kbound$ be zeros coincides with the claim of the proposition.
    
    We can do this by an argument similar to that in the proof of Proposition~\ref{prop:tconst}.  
    We consider the case $\alpha = \alpha_1 \alpha_2$ (the other cases are similar).
    Take $p_1, \dots, p_\kbound$.  Fix $\sconst_2$ by the induction hypothesis for $\alpha_2$ over $p_1,\dots,p_\kbound$ and $\sconst_1$ by the induction hypothesis for $\alpha_1$ over $\sconst_2, \dots, \sconst_2$ ($\kbound$ times).  Introduce $v \in \reflang(\alpha), k$ and $i$.  Decompose $v = v_1 v_2 = v_1 v_{2,1} \cdots v_{2,\zettaiti{v_2}}$ where $v_l \in \reflang(\alpha_l)$ for $l = 1, 2$ and $v_{2,1}, \dots, v_{2,\zettaiti{v_2}} \in \nfaalph$.  
    Because $\myopen(v_1) = \emptyset$ by Corollary~\ref{cor:finalopisempty}, $\myopen(v_1 v_{2,1} \cdots v_{2,i}) = \myopen(v_{2,1} \cdots v_{2,i})$ and thus $s^{v,(p_1,\dots,p_\kbound)}_{\zettaiti{v_1}+i,k} = s^{v_2,(p_1,\dots,p_\kbound)}_{i,k} \leq \sconst_2$ for $i = 0, 1, \dots, \zettaiti{v_2}$ by the induction hypothesis for $\alpha_2$.  In particular, each $p'_k \leq \sconst_2$ where $(p'_1,\dots,p'_\kbound) = s^{v_2,(p_1,\dots,p_\kbound)}_0$.
    From this, observe by easy induction on $i$ that $s^{v_1,(p'_1,\dots,p'_\kbound)}_{i,k} \leq s^{v_1,(\sconst_2,\dots,\sconst_2)}_{i,k}$ for $i = 0,1,\dots,\zettaiti{v_1}$.  Therefore, we obtain $s^{v,(p_1,\dots,p_\kbound)}_{i,k} = s^{v_1,(p'_1,\dots,p'_\kbound)}_{i,k} \leq s^{v_1,(\sconst_2,\dots,\sconst_2)}_{i,k} \leq \sconst_1$ by the induction hypothesis for $\alpha_1$.  Take $\sconst = \max\syuugou{\sconst_1, \sconst_2}$.
\end{proof}
Fix $\tconst$ and $\sconst$ in Propositions~\ref{prop:tconst} and \ref{prop:sconst}.  We can assume $\tconst, \sconst \geq 1$.
Define $\mydepth_l$ inductively as follows: $\mydepth_0 = 0$ and $\mydepth_l = \tconst (1 + \kbound\cdot \mydepth_{l-1} )$.  Let $\rbound = \mydepth_{\sconst}$.
Fix an NFA $\myautomaton_{\alpha} = (Q,\nfaalph,\delta,q_0,F)$ equivalent to $\reflang(\alpha)$ satisfying the NFA conditions (a) and (b) given in the paragraph immediately above Lemma~\ref{lem:om} in Section~\ref{sec:rewbpmcfg}.  Let $n = \zettaiti{Q}$ and $Q = \{q_0,\dots,q_{n-1}\}$.
Let $G'_\alpha$ be a $(\kbound \cdot \rbound+1)$-unary-MCFG $(\mynonterminals,\Sigma,\mathcal{F},P,S)$:
$\mynonterminals$ consists of the initial nonterminal $S$ and nonterminals $A_i^{\tau,\sigma}$ for all $0 \leq i < n, 1 \leq k \leq \kbound, \tau = (\tau_1,\dots,\tau_\kbound)\;(0 \leq \tau_1,\dots,\tau_\kbound \leq \tconst)$ and $\sigma = (\sigma_1,\dots,\sigma_\kbound)\;(0 \leq \sigma_1,\dots,\sigma_\kbound \leq \sconst)$.
$\mathcal{F}$ consists of $\symo{},\symi_{a}^{M},\symp_{k}^{M,\tau,\sigma},\symr_{k}$ and $\symid$ for all $a \in \Sigma, 1 \leq k \leq \kbound, M \subseteq \syuugou{1,\dots,\kbound}, \tau = (\tau_1,\dots,\tau_\kbound)\;(0 \leq \tau_1,\dots,\tau_\kbound \leq \tconst, \tau_k \geq 1)$ and $\sigma = (\sigma_1,\dots,\sigma_\kbound)\;(0 \leq \sigma_1,\dots,\sigma_\kbound \leq \sconst, \sigma_k \geq 1)$.  Here,
\begin{itemize}
    \item $\symo = \lambda (x,(y_1),\dots,(y_\kbound)).x$, $\symid = \lambda(x,(y_1),\dots,(y_\kbound)).(x,(y_1),\dots,(y_\kbound))$,
    \item $\symi_{a}^{M} = \lambda(x,(y_1),\dots,(y_\kbound)).(xa,(z_1),\dots,(z_\kbound))$ where $z_{j,r} = y_{j,r} a$ $(1 \leq r \leq \rbound)$ for $j \in M$ and $(z_j) = (y_j)$ for $j \notin M$.
    \item Let $\psi(k,0) = ((\tau_k-1)/\tconst)\mydepth_{\sigma_k} + 1$ and $\psi(k,j) = \psi(k,0) + j\mydepth_{\sigma_{k}-1}$ for $1 \leq k,j \leq \kbound$.  Then, 
        \[
            \symp_{k}^{M,\tau,\sigma} = \lambda(x,(y_1),\dots,(y_\kbound)).(xy_{k,\psi(k,0)},(z_1),\dots,(z_\kbound)),
        \]
        where in the case $j \in M$ and $j \neq k$, $z_{j,r} = y_{j,r} y_{k,\psi(k,j-1)+r}$ for $1 \leq r \leq \mydepth_{\sigma_k -1}$ and $z_{j,r} = \varepsilon$ for $\mydepth_{\sigma_k -1} < r \leq \rbound$; in the case $j \notin M$ and $j \neq k$, $(z_{j}) = (y_{j})$; $z_{k,r} = y_{k,r}$ for $1 \leq r < \psi(k,0)$ and $z_{k,r} = \varepsilon$ for $\psi(k,0) \leq r \leq \rbound$,
    \item $\symr_{k}^M = \lambda(x,(y_1),\dots,(y_\kbound)).(x,(z_1),\dots,(z_\kbound))$ where $(z_{k}) = (\varepsilon) = (\varepsilon, \dots, \varepsilon)$ ($\rbound$ times) and $(z_{j}) = (y_{j})$ for $j \neq k$.
\end{itemize}
As before, we shall denote $\varphi_{a}^{M,\tau,\sigma} = \symi_{a}^{M}$, $\varphi_{k}^{M,\tau,\sigma} = \symp_{k}^{M,\tau,\sigma}$, $\varphi_{\lbrack_k}^{M,\tau,\sigma} = \symr_{k}$ and $\varphi_{\rbrack_k}^{M,\tau,\sigma} = \symid$.
$P$ consists of (1) $S \to \symo A_f^{(0,\dots,0),(0,\dots,0)}$ for all $q_f \in F$, (2) $A_j^{\tau',\sigma'} \to \varphi_c^{M_j,\tau,\sigma} A_i^{\tau,\sigma}$ for all $q_i \myto{c}{\myautomaton_\alpha} q_j, \tau \tarrow{c} \tau'$ and $\sigma \sarrow{c} \sigma'$ (refer to Definitions~\ref{def:t} and \ref{def:s} for the dashed arrows), and (3) $A_0^{\tau,\sigma} \to (\varepsilon, (\varepsilon), \dots, (\varepsilon))$ for all $1 \leq k \leq \kbound, \tau = (\tau_1,\dots,\tau_\kbound)\;(0 \leq \tau_1,\dots,\tau_\kbound \leq \tconst)$ and $\sigma = (\sigma_1,\dots,\sigma_\kbound)\;(0 \leq \sigma_1,\dots,\sigma_\kbound \leq \sconst)$.
The following lemma, analogous to Lemma~\ref{lem:base}, holds.

\begin{lemma}
    \label{lem:mcfgbase}
    Let $v = v_1 \cdots v_{\zettaiti{v}} \in \reflang(\alpha)$ where $v_1, \dots, v_{\zettaiti{v}} \in \nfaalph$.  Then, for every $i \geq 0$,
        \[
            \varphi_{v_i}^{\myopen(v_1\cdots v_i),t_{i-1},s_{i-1}} \cdots \varphi_{v_1}^{\myopen(v_1),t_{0},s_{0}} (\varepsilon,(\varepsilon),\dots,(\varepsilon)) = (\deref(v_1\dots v_i),(\nu_1^{i}),\dots,(\nu_\kbound^{i}))
        \]
    where for each $1 \leq k \leq \kbound$, $\nu_{k,r}^{i} = \mu_k(v_1 \cdots v_i)$ holds for $1 \leq r \leq (t_{i,k}/\tconst) \mydepth_{s_{i,k}}$.
    In particular,
    \[
        \symo\cdot\varphi_{v_{\zettaiti{v}}}^{\myopen(v_1\cdots v_{\zettaiti{v}}),t_{\zettaiti{v}-1},s_{\zettaiti{v}-1}} \cdots \varphi_{v_1}^{\myopen(v_1),t_0,s_0} (\varepsilon, (\varepsilon), \dots, (\varepsilon)) = \deref(v).
    \]
\end{lemma}

    Before the proof, we recall the following notation and facts from our previous study~\cite{mfcs2023,ic2025}.  
    A string $v$ over $\Sigma \mydisjointu \brkset \mydisjointu \mynat$ is \emph{matching} when for all $i \in \mynat$ that appears in $v$, if $\mathord{\lbrack_i}$ appears in $v$ before this $i$, then $\mathord{\rbrack_i}$ appears between these $\mathord{\lbrack_i}$ and $i$.
    We cite the following facts with respect to the matching property from the same paper.
\begin{claim} \label{clm:mat}
    \begin{enumerate}[(i)]
        \item Every $v \in \reflang(\alpha)$ is matching.
        \item A prefix of a matching string is matching.
    \end{enumerate}
\end{claim}

\begin{proof}[Proof of Lemma~\ref{lem:mcfgbase}]
    We prove by induction on $i$.  The base case $i = 0$ is obvious.  
    Suppose the inductive case $i \Longrightarrow i+1$.  By the induction hypothesis, it suffices to show
    \[
        \varphi_{v_{i+1}}^{\myopen(v_1\cdots v_{i+1}),t_{i},s_{i}} (\deref(v_1\cdots v_{i}), (\nu_1^{i}), \dots, (\nu_\kbound^{i})) = (\deref(v_1\cdots v_{i+1}), (\nu_1^{i+1}), \dots, (\nu_\kbound^{i+1}))
    \]
    where all $(\nu_1^{i+1}), \dots, (\nu_\kbound^{i+1})$ satisfy the conditions of the claim.
    In the cases $v_{i+1} = a \in \Sigma$ or $\rbrack_k$, it clearly holds by Proposition~\ref{prop:memind} because $t_{i+1} = t_{i}$ and $s_{i+1} = s_{i}$ from those definitions.
    In the case $v_{i+1} = \lbrack_k$, for $j \neq k$, we have $t_{i+1,j} = t_{i,j}$ and $s_{i+1,j} = s_{i,j}$, and thus $\nu_j^{i+1} = \nu_j^i$ satisfies the conditions of the claim by Proposition~\ref{prop:memind}.
    For $j = k$, $\symr_{k}$ resets all $\nu_{k,r}^i$ $(1 \leq r \leq \rbound)$ to $\varepsilon$'s and therefore $(\nu_k^{i+1}) = (\varepsilon)$ as required.
    The remaining case is $v_{i+1} = k$ $(1 \leq k \leq \kbound)$.  Let $\psi_i(k,0) = ((t_{i,k}-1)/\tconst)\mydepth_{s_{i,k}} + 1$ and $\psi_i(k,j) = \psi_i(k,0) + j\mydepth_{s_{i,k}-1}$ for $1 \leq k,j \leq \kbound$.
    Note that $t_{i,k}, s_{i,k} \geq 1$ since $v_{i+1} = k$.
    Because $\deref(v_1\cdots v_i k) = \deref(v_1 \cdots v_i) \mymem_k(v_1 \cdots v_i) = \deref(v_1 \cdots v_i) \nu_{k,\psi_i(k,0)}^{i}$, the first components of both sides of the equation coincide.  It remains to show that, for any $1 \leq j \leq \kbound$, $\nu_{j,r}^{i+1} = \mymem_j(v_1 \cdots v_i k)$ $(1 \leq r \leq (t_{i+1,j}/\theta)\mydepth_{s_{i+1,j}})$.  Take any $1 \leq j \leq \kbound$.
    \begin{itemize}
        \item If $j \in \myopen(v_1\cdots v_i)$ and $j\neq k$, then $\mymem_j(v_1 \cdots v_i k) = \mymem_j(v_1 \cdots v_i) \mymem_k(v_1 \cdots v_i)$ by Proposition~\ref{prop:memind}.  
        Note that $s_{i,k} > s_{i+1,j}$ from Definition~\ref{def:s} and $t_{i+1,j} \leq \tconst$.  Thus, for $1 \leq r \leq (t_{i+1,j}/\tconst)\mydepth_{s_{i+1,j}}$, $\psi_i(k,j-1)+r \leq \psi_i(k,\kbound-1)+\mydepth_{s_{i,k}-1} = (t_{i,k}/\tconst) \mydepth_{s_{i,k}}$.  
        Hence, for $1 \leq r \leq (t_{i+1,j}/\tconst)\mydepth_{s_{i+1,j}}$, we have $\nu_{j,r}^{i+1} = \nu_{j,r}^{i} \nu_{k,\psi_i(k,j-1)+r}^{i} = \mymem_j(v_1\cdots v_i) \mymem_k(v_1 \cdots v_i)$ by the induction hypothesis.  
        \item If $j \notin \myopen(v_1 \cdots v_i)$ and $j \neq k$, then $\mymem_j(v_1 \cdots v_i k) = \mymem_j(v_1 \cdots v_i)$ by Proposition~\ref{prop:memind}.  
        Because $t_{i+1,j} = t_{i,j}$ and $s_{i+1,j} \leq s_{i,j}$ from Definitions~\ref{def:t} and \ref{def:s}, and $(\nu_j^{i+1}) = (\nu_j^{i})$, we have $\nu_{j,r}^{i+1} = \nu_{j,r}^{i} = \mymem_j(v_1 \cdots v_i)$ for $1 \leq r \leq (t_{i+1,j}/\tconst) \mydepth_{s_{i+1,j}} \leq (t_{i,j}/\tconst) \mydepth_{s_{i,j}}$ by the induction hypothesis.
    \item If $j = k$, $v_1 \cdots v_{i}k$ is matching from Claim~\ref{clm:mat} and thus $k \notin \myopen(v_1 \cdots v_i).$
            Thus, $\mymem_k(v_1 \cdots v_i k) = \mymem_k(v_1 \cdots v_i)$ by Proposition~\ref{prop:memind}.
            Because $t_{i+1,k} = t_{i,k} - 1$ and $s_{i+1,k} \leq s_{i,k}$ from Definitions~\ref{def:t} and \ref{def:s}, we have $\nu_{k,r}^{i+1} = \nu_{k,r}^{i} = \mymem_k(v_1\cdots v_i)$ for $1 \leq r \leq (t_{i+1,k}/\tconst) \mydepth_{s_{i+1,k}} \leq ((t_{i,k}-1)/\tconst) \mydepth_{s_{i,k}} < \psi_i(k,0)$ by the induction hypothesis.
    \end{itemize}
    This completes the proof.
\end{proof}

\begin{proof}[Proof of Theorem~\ref{thm:rewbscmcfg}]
    We give a proof of $L(G'_\alpha) \subseteq L(\alpha)$ (The other direction is quite similar to the proof of that direction of Theorem~\ref{thm:rewbpmcfg}).
    Take any $w \in L(G'_\alpha)$ and fix a derivation 
    \begin{align*}
        S &\to \symo A_{\myperm_m}^{\tau_{(m)},\sigma_{(m)}} \to \symo\cdot f_m^{M_{\myperm_m},\tau_{(m-1)},\sigma_{(m-1)}} A_{\myperm_{m-1}}^{\tau_{(m-1)},\sigma_{(m-1)}} \to \cdots \\
        &\to \symo\cdot f_m^{M_{\myperm_m},\tau_{(m-1)},\sigma_{(m-1)}} \cdots f_1^{M_{\myperm_1},\tau_{(0)},\sigma_{(0)}} A_{\myperm_0}^{\tau_{(0)},\sigma_{(0)}} \\ 
        &\to \symo\cdot f_m^{M_{\myperm_m},\tau_{(m-1)},\sigma_{(m-1)}} \cdots f_1^{M_{\myperm_1},\tau_{(0)},\sigma_{(0)}} (\varepsilon,(\varepsilon),\dots,(\varepsilon)) = w.
    \end{align*}
    For each $i = 1, \dots, m$, $G'_\alpha$ has the rule $A_{\myperm_i}^{\tau_{(i)},\sigma_{(i)}} \to f_{i}^{M_{\myperm_i},\tau_{(i-1)},\sigma_{(i-1)}} A_{\myperm_{i-1}}^{\tau_{(i-1)},\sigma_{(i-1)}}$.  From the construction of $G'_\alpha$, there exists $v_i$ such that $\varphi_{v_i}^{M_{\myperm_i},\tau_{(i-1)},\sigma_{(i-1)}} = f_{i}^{M_{\myperm_i},\tau_{(i-1)},\sigma_{(i-1)}}$, $q_{\myperm_{i-1}} \myto{v_i}{} q_{\myperm_i}$, $q_{\myperm_0} = q_0$, $q_{\myperm_m} \in F$, and $\tau_{(m)} = (0,\dots,0), \sigma_{(m)} = (0,\dots,0), \tau_{(i-1)} \tarrow{v_i} \tau_{(i)}, \sigma_{(i-1)} \sarrow{v_i} \sigma_{(i)}$.
    In particular, we have $v = v_1 \cdots v_m \in \reflang(\alpha)$ and $\tau_{(i)} = t_i, \sigma_{(i)} = s_i$ $(0 \leq i \leq m)$.
    Because $q_0 \mytto{v_1 \cdots v_i}{} q_{\myperm_i}$, by Lemma~\ref{lem:om}, $M_{\myperm_i}$ coincides with $\myopen(v_1\cdots v_i)$.
    Therefore, $w = \deref(v)$ by Lemma~\ref{lem:mcfgbase}.
\end{proof}
With Theorem~\ref{thm:notmcfg}, we have:
\begin{corollary}
    \label{cor:scsmaller}
    The language class of closed-star rewbs is a proper subclass of rewbs.
\end{corollary}

\section{Closed-star rewbs and NESLs}
\label{sec:rewbscnesa}

Ginsburg et al.~\cite{ginsburg1967stack, ginsburg1967one} proposed \emph{stack automata} (SA), pushdown automata (PDA) extended with the ability to read non-top-most positions of the stack.  \emph{Nonerasing stack automata} (NESA) are SA but with stack popping prohibited~\cite{ginsburg1967stack,ogden1969intercalation}, and \emph{nested stack automata} (Nested-SA) are SA extended with the ability to create stacks nested inside other stacks~\cite{aho1968indexed,aho1969nested}.
As stated in the introduction, these papers have shown that IL, the language class corresponding to Nested-SA, is properly contained in CSL; SL, the language class corresponding to SA, is properly contained in IL; and NESL, the language class corresponding to NESA, is properly contained in SL: $\text{NESL} \subsetneq \text{SL} \subsetneq \text{IL} \subsetneq \text{CSL}$.  Also as stated there, our recent paper~\cite{mfcs2023,ic2025} has shown that the language class of rewbs is a proper subclass of ILs but not a subclass of SLs, and that that of rewbs without captured references is a proper subclass of NESLs.
    As we have shown, the language class of rewbs without captured references is not a subclass of MCFLs (Theorem~\ref{thm:notmcfg}), but that of closed-star rewbs is (Theorem~\ref{thm:rewbscmcfg}).  
    The aim of this section is to prove that the language class of closed-star rewbs also falls inside NESL:
    
    \begin{restatable}{theorem}{thmrewbscisnesa} \label{thm:rewbscisnesa}
        Every closed-star rewb describes an NESL.
    \end{restatable}

    We informally explain the proof idea.
Roughly, NESA can record information about the preceding process in its nonerasing stack (i.e., list of increasing length) and later review it by moving its stack pointer.  \cite{mfcs2023,ic2025} shows an explicit construction of a Nested-SA that is equivalent to the given rewb.
The proof we shall present is similar to this idea in recording a string from the ref-language of the given rewb in its stack memory and simulating a dereferencing process by stack pointer manipulations.

        The main difference is how we bookmark where to come back after a dereference within another dereference has finished.
        Although the construction of \cite{mfcs2023,ic2025} utilized a nested stack as a bookmarker for this purpose, it would not work for the current case because NESA does not have the ability to create a nested stack.    
In fact, as remarked above, there exists a rewb language that is not an SL, much less an NESL.  Thus, there is no NESA construction that works for all rewbs.
However, such a construction is possible for closed-star rewbs.

Our NESA does the bookmarking by implementing a ``call-stack'' for fetching bracketed strings in dereferencing by its states.
Roughly, when simulating a replacement of a number character $i$, the NESA (i) extends the call-stack, (ii) refers to the corresponding bracketed string $\lbrack_i \gamma \rbrack_i$ in its nonerasing stack, (iii) handles $\gamma$ recursively, and (iv) reverts the call-stack.
The subtle point in (ii) is that the NESA counts the number of occurrences of $i$ between the brackets and the starting $i$ and records it in the call-stack-top so that it can come back to the starting $i$ after handling $\gamma$ at (iii).  Namely, the count plays the role of a bookmarker.
The closed-star condition ensures that both the length of the call-stack and the count are bounded thus making this construction possible.

We explain the formal details.
A (one-way\footnote{That is, the input cursor does not move back to left.} and nondeterministic) SA $A$ is a 9-tuple $(Q,\Sigma,\Gamma,\delta,q_0,Z_0,\inputendmarker,\mydollar,F)$ satisfying the following conditions.
    The components $Q$, $\Sigma$, $q_0$ and $F$ are the same as those of NFA (Definition~\ref{def:nfa}).
    Besides, $\Gamma$ is a finite set of stack symbols and $Z_0 \in \Gamma$ is an initial stack symbol.  The special symbol $\inputendmarker \notin \Sigma$ denotes the endmarker of the input tape and symbol $\mydollar \notin \Gamma$ denotes the top of the stack.
    The transition function $\delta$ has two modes, where $\myleft,\mystop,\myright \notin (\Sigma\cup\Gamma) \mydisjointu \syuugou{\mydollar}$, $\dir_i = \syuugou{\mystop,\myright}$, and $\dir_s = \syuugou{\myleft,\mystop,\myright}$:
    
    \begin{enumerate}[(i)]
        \item (pushdown mode) $Q\times (\Sigma \mydisjointu \{ \inputendmarker \}) \times \Gamma\mydollar \to \powerset({Q\times \dir_i \times \Gamma^\ast \mydollar})$,
        \item (stack reading mode) $Q \times (\Sigma \mydisjointu \{ \inputendmarker \}) \times (\Gamma\mydollar \mydisjointu \Gamma) \to \powerset({Q\times \dir_i\times \dir_s})$.
    \end{enumerate}
    
	An NESA is an SA whose transition function $\delta$ satisfies the condition that, in (i) (pushdown mode), $\delta(q,a,Z\mydollar) \ni (q^\prime, d, w\mydollar)$ implies $w \neq \varepsilon$.
    Let $\mydtoz{\myleft}=-1$, $\mydtoz{\mystop}=0$ and $\mydtoz{\myright}=1$.
    %
        We define $I = Q \times \Sigma^\ast \{ \inputendmarker \} \times (\Gamma \mydisjointu \{ \,\mathord{\myspl} \})^\ast \{\mydollar\}$, the set of \emph{instantaneous descriptions} (ID).  Here, $\myspl\,\notin \Gamma$ stands for the position of stack pointer.  Also, let $\myid{A}$ (or $\myid{}$ whenever $A$ is clear) be the smallest binary relation over $I$ satisfying the following conditions for all $q, q' \in Q$, $k \geq 0$, $a_0, \dots, a_k \in \Sigma \mydisjointu \{ \inputendmarker \}$, $Z, Z_1, \dots, Z_n \in \Gamma$, $\gamma, y \in \Gamma^\ast$, $d \in \dir_i$ and $e \in \dir_s$ with $k = 0 \rightarrow d \neq \myright$:

        \begin{enumerate}[(i)]
            \item if $\delta(q,a_0,Z\mydollar) \ni (q^\prime, d, y\mydollar)$, $(q,a_0 \cdots a_k, \gamma Z\myspl\!\mydollar) \myid{A} (q^\prime, a_{\mydtoz{d}} \cdots a_k, \gamma y \myspl\!\mydollar)$,
            \item if $\delta(q,a_0,\square) \ni (q^\prime, d, e)$ where $\square \in \{ Z_j \mydollar, Z_j \}$ and $1 \leq j \leq n$ with 
                $j = 1 \to e \neq \myleft$ and $j = n \to e \neq \myright$, $(q,a_0 \cdots a_k, Z_1 \cdots Z_{j} \myspl\! \cdots Z_n \mydollar) \myid{A} (q^\prime,a_{\mydtoz{d}} \cdots a_k, Z_1 \cdots Z_{j+\mydtoz{e}} \myspl\! \cdots Z_n \mydollar)$.
		\end{enumerate}

        Note that $\myleft \notin \dir_i$ asserts that it is one-way.
        We say that $A$ accepts $w \in \Sigma^\ast$ if there exist $\gamma_1, \gamma_2 \in \Gamma^\ast$ and $q_f \in F$ such that $(q_0, w \inputendmarker, Z_0 \myspl\!\mydollar) \myid{A}^\ast (q_{f}, \inputendmarker, \gamma_1 \myspl\! \gamma_2\, \mydollar)$. Let $L(A)$ denote the set of all strings accepted by $A$.


Let $\alpha$ be a closed-star rewb.  Hereafter, we shall present our NESA construction $A_\alpha$ equivalent to $L(\alpha)$.
Fix an NFA $\myautomaton = (Q_\myautomaton,\nfaalph,\delta_\myautomaton,q_0,F)$ equivalent to $\reflang(\alpha)$ whose every state can reach some final state.
        Also, fix $\tconst$ and $\sconst$ as in the previous section.
Let $A_\alpha$ be the NESA $(Q,\Sigma,\Gamma,\delta,q_0,Z_0,\inputendmarker,\mydollar,F)$ defined as follows.  Let $\cmdset = \{ \symcall_i^{x}, \symexec_i^{x}, \symret_i^{x} \mid i \in \{ 1, \dots, \kbound \}, x \in \{ 0, \dots, \tconst \} \}$.  The state set $Q$ is $\{ q\xi \mid q \in Q_\myautomaton, \xi \in \cmdset^{\leq \sconst}\}$, where $\cmdset^{\leq \sconst}$ is the set of all strings over $\cmdset$ of length less than or equal to $\sconst$.  Note that $Q$ is finite.  The stack symbol set $\Gamma$ is $\nfaalph \mydisjointu \{Z_0\}$.  The transition function $\delta$ is defined as the smallest relation satisfying the following 10 rules for all $q, q' \in Q_\myautomaton$, $a \in \Sigma$, $b \in \{ \lbrack_i, \rbrack_i \mid i \in \{ 1, \dots, \kbound \} \}$, $c \in \nfaalph \mydisjointu \{ \inputendmarker \}$, $Z \in \Gamma$, $i \neq j \in \{ 1, \dots, \kbound \}$, $x \in \{ 0, \dots, \tconst \}$ and $\xi \in \cmdset^{\leq \sconst}$.  

\begin{figure}[htb]
    \begin{minipage}[t]{0.54\linewidth}
        \begin{enumerate}[(1)]
    \item $\delta_\myautomaton(q,a) \ni q' \Rightarrow \delta(q,a,Z\mydollar) \ni (q', \myright, Za\mydollar)$
    \item $\delta_\myautomaton(q,b) \ni q' \Rightarrow \delta(q,c,Z\mydollar) \ni (q', \mystop, Zb\mydollar)$
    \item $\delta_\myautomaton(q,i) \ni q' \Rightarrow \delta(q,c,Z\mydollar) \ni (q' \symcall_i^{0}, \mystop, Zi \mydollar)$
    \item $\delta(q \xi \symcall_i^{x}, c, \square)$ \\
        $= \left\{
        \begin{array}{ll}
            \{ (q \xi \symcall_i^{x+1}, \mystop, \myleft) \} & \mytumeru (x < \tconst, \square \in \{ i \mydollar, i \}) \\
            \{ (q \xi \symexec_i^{x}, \mystop, \myright) \} & \mytumeru (\square = \lbrack_i) \\
            \{ (q \xi \symret_i^{x}, \mystop, \myright) \} & \mytumeru (\square = Z_0) \\
            \{ (q \xi \symcall_i^{x}, \mystop, \myright) \} & \mytumeru (otherwise)
        \end{array}
        \right.$ 
    \item $\delta(q \xi \symexec_i^{x}, a, a) = \{ (q \xi \symexec_i^{x}, \myright, \myright) \}$
\end{enumerate}
    \end{minipage}
    \begin{minipage}[t]{0.45\linewidth}
    \begin{enumerate}[(1)]
    \setcounter{enumi}{5}
    \item $\delta(q \xi \symexec_i^{x}, c, \mathrm{bk}_j) = \{ (q \xi \symexec_i^{x}, \mystop, \myright) \}$ \\
        where $\text{bk} \in \syuugou{\mathord{\lbrack},\mathord{\rbrack}}$
    \item $\delta(q \xi \symexec_i^{x}, c, j) = \{ (q \xi \symexec_i^{x} \symcall_j^{0}, \mystop, \mystop) \}$       
    \item $\delta(q \xi \symexec_i^{x}, c, \rbrack_i) = \{ (q \xi \symret_i^{x}, \mystop, \myright) \}$
    \item $\delta(q \xi \symret_i^{x}, c, \square)$ where $x > 0$ \\
        $= \left\{
         \begin{array}{ll}
             \{ (q \xi \symret_i^{x-1}, \mystop, \mystop) \} & \mytumeru (\square = i \mydollar) \\
             \{ (q \xi \symret_i^{x-1}, \mystop, \myright) \} & \mytumeru (\square = i) \\
             \{ (q \xi \symret_i^{x}, \mystop, \myright) \} & \mytumeru (otherwise)
         \end{array}
         \right.$
     \item $\delta(q \xi \symret_i^{0}, c, \square) = \{ (q \xi, \mystop, \mystop) \}$ \\
        where $\square \in \{ Z \mydollar, Z \}$
\end{enumerate}
    \end{minipage}
    %
\end{figure}

The constructed NESA $A_\alpha$ consumes $w \in L(\alpha)$ while tracing a path in the NFA $\myautomaton$ and guessing some $v \in \reflang(\alpha)$ such that $\deref(v) = w$ character by character.
        When guessing $a \in \Sigma$, $A_\alpha$ consumes $a$ from $w$ and pushes $a$ to the stack by rule (1) (\emph{$a$-rule}).  Guessing $b \in \brkset$, $A_\alpha$ only pushes $b$ to the stack by rule (2) (\emph{$b$-rule}).  Guessing $i \in \mynat$, $A_\alpha$ pushes $i$ and then consumes $\mymem_i(\gamma)$ where the stack is $Z_0 \gamma \myspl\!\mydollar$ at that time, namely the string which should be consumed in the dereference of this $i$.  This is done by a big ``transition'' that consists of several transitions that use rules (3)--(10) (\emph{$i$-rules}).  The most significant part of the proof is verifying the correctness of this big-step semantics of $i$-rules (Lemma~\ref{lem:irules} below).  

    Specifically, $A_\alpha$ simulates the dereference of a guessed $i$ as follows.
        $A_\alpha$ regards its states as the snapshots of a call-stack of length $\eta + 1$.  In what follows, we call the last character of a state a \emph{state top}.
    In the ``transition,'' $A_\alpha$ decides its next move in accordance with the state top, the input character, and the stack symbol indicated by the stack pointer.  
    
    Let us assume that $A_\alpha$ is at state $q$.  $A_\alpha$ first pushes the guessed $i$ to its stack and goes to state $q \symcall_i^{0}$.
    Then, the first and last branches in rule (4) move its stack pointer downwardly until meeting $\mathord{\lbrack_i}$ (or $Z_0$ if it does not exist). During this process, the superscript $x$ of the state top $\symcall_i^{x}$ is incremented every time $i$ appears.  For brevity, we assume that $\mathord{\lbrack_i}$ was found.  Consequently, when the descending move has finished, the stack pointer indicates the most recent $\mathord{\lbrack_i}$ and the superscript $x$ of the state top stores the number of occurrences of $i$ between this $\mathord{\lbrack_i}$ and the starting $i$.
    Remark that the remaining tasks are (i) to consume the string between the $\mathord{\lbrack_i}$ and the right nearest $\mathord{\rbrack_i}$ from the input (or to consume nothing if it falls down to $Z_0$) and (ii) to come back to the starting $i$.
    
    For task (i), the second branch in rule (4) ``switches'' the state top from $\symcall_i^{x}$ to $\symexec_i^{x}$.
    Rules (5)--(8) stipulate the behavior of $A_\alpha$ when the state top is $\symexec_i^{x}$.  Roughly, while moving its stack pointer upward, $A_\alpha$ skips brackets except $\rbrack_i$ (rules (6) and (8)), and compares the input symbol with the pointed stack symbol whenever it is in $\Sigma$, making progress by consuming the input symbol only if the two symbols coincide (rule (5)).  If a number character $j$ was found, it ``pushes'' $\symcall_j^{0}$ to its state and starts with rule (4) again to simulate the dereference of this $j$ (rule (7)).  If $\mathord{\rbrack_i}$ was found, rule (8) ``switches'' the state top from $\symexec_i^{x}$ to $\symret_i^{x}$ for task (ii) mentioned above.
    In contrast with the case of $\symcall_i^{x}$, $A_\alpha$ now climbs up the stack while decrementing the superscript $x$ every time $i$ appears in its way (rule (9)).  Recall what $x$ had stored.  Also, by the definition of rewbs (Definition~\ref{def:rewb}), no $i$ appears between the $\mathord{\lbrack_i}$ and the $\mathord{\rbrack_i}$.  Therefore, the position to come back is where the superscript becomes zero (rule (10)).
    
    Note that this big ``transition'' is designed to be one-way, namely the applicable rule, if any, is uniquely determined from an ID being calculated within the ``transition.'' 
    We refer to Proposition~\ref{prop:det} for more details.

    For each $n$, we write $\mathord{\myid{(n)}}$ for the subrelation of $\mathord{\myid{}}$ derived from rule $(n)$.
    Below, the two lemmas characterize the semantics of $a$-rule, $b$-rule and $i$-rules.    
    \begin{lemma} \label{lem:abrule}
        Let $q,q' \in Q_\myautomaton$, $w,w' \in \Sigma^\ast$, $\gamma \in \Gamma^\ast$ and $\beta \in (\Gamma \mydisjointu \{\, \myspl \})^\ast$.  
        Then,
        \begin{enumerate}[(a)]
            \item $(q,w\inputendmarker,Z_0 \gamma \myspl\! \mydollar) \myid{(1)} (q',w'\inputendmarker,\beta \mydollar) \iff \exists a \in \Sigma.\, q \myto{a}{\myautomaton} q' \land w = aw' \land \beta = Z_0 \gamma a \myspl$,
            \item $(q,w\inputendmarker,Z_0 \gamma \myspl\! \mydollar) \myid{(2)} (q',w'\inputendmarker,\beta \mydollar) \iff \exists b \in \brkset.\, q \myto{b}{\myautomaton} q' \land w = w' \land \beta = Z_0 \gamma b \myspl$.
        \end{enumerate}
    \end{lemma}
\begin{lemma} \label{lem:irules}
    Let $q,q',p \in Q_\myautomaton$, $w,w' \in \Sigma^\ast$, $\gamma \in \Gamma^\ast$ and $\beta \in (\Gamma \mydisjointu \{\, \myspl \})^\ast$.  
    Suppose that $q_0 \mytto{\gamma}{\myautomaton} q$.
    For each $i \in \mynat$, the following equivalence holds:
    $(q,w\inputendmarker,Z_0 \gamma \myspl\! \mydollar) \myid{(3)} (q' \symcall_i^{0}, w\inputendmarker, Z_0 \gamma i \myspl\! \mydollar) \myid{} \cdots \myid{} (p,w'\inputendmarker,\beta \mydollar)$ where no ID with a state in $Q_\myautomaton$ appears in the intermediate steps $\cdots$ $\iff q \myto{i}{\myautomaton} q' \land p = q' \land w = \mymem_i(\gamma) w' \land \beta = Z_0 \gamma i \myspl\!$.
\end{lemma}

Lemma~\ref{lem:abrule} is immediate.
To prove Lemma~\ref{lem:irules}, 
we recall from \cite{mfcs2023,ic2025} the set $I_\bot = Q \times (\Sigma^\ast \{ \inputendmarker \} \mydisjointu \{ \bot \}) \times (\Gamma \mydisjointu \{ \,\mathord{\myspl} \})^\ast \{ \mydollar \}$ of IDs expanded with a distinguished symbol $\bot \notin \Sigma \mydisjointu \{ \inputendmarker \}$ and the relation $\myonlymove{(n)}$ over $I_\bot$ defined as $C(u) \myonlymove{(n)} C'(u')$ for $u, u' \in \Sigma^\ast \{ \inputendmarker \} \mydisjointu \{ \bot \}$ if and only if $u' = \bot$ or there exist $w, w' \in \Sigma^\ast$ such that $u = w \inputendmarker \land u' = w' \inputendmarker \land \forall j, C''. C(u) \myid{(j)} C'' \iff (j = n \land C'' = C'(u'))$.\footnote{A notation like $C(u)$ stands for an ID whose input string is $u$.  The $(u)$ is omitted if not needed.}  We often omit the subscript $(n)$ and simply write $\myonlymove{}$.  The following claim is immediate from the definition.

    \begin{claim}  Let $w, w' \in \Sigma^\ast$ and $n \in \mynat$.
        \begin{enumerate}[(i)]
            \item If $C(w \inputendmarker), C'(w' \inputendmarker) \in I$, then $C(w \inputendmarker) \myid{(n)} C'(w' \inputendmarker) \iff C(w \inputendmarker) \myonlymove{(n)} C'(w' \inputendmarker)$.
            \item Suppose that $C(w \inputendmarker) \myonlymove{} C'(u)$.  For all $C''$, $C(w \inputendmarker) \myid{(n)} C''$ implies $u \neq \bot \land C'' = C'(u)$.
        \end{enumerate}
    \end{claim}

Let $w \backslash w_1$ denote $w_2$ if $w = w_1 w_2$, and $\bot$ otherwise.

\begin{proposition} \label{prop:det}
    Let $q' \in Q$, $\xi \in \cmdset^\ast$, $w \in \Sigma^\ast$, $\gamma, \gamma' \in \Gamma^\ast$ and $i \in \mynat$.
    Suppose that (a) $q' \xi \symcall_i^{0} \in Q$, (b) $y = \gamma i \gamma'$ is a prefix of some $v \in \reflang(\alpha)$, and (c) $s_{\zettaiti{\gamma},i}^{y} \geq \zettaiti{\xi} + 1$ (see Definition~\ref{def:s} for $s$).  Then,
    \[
        (q' \xi \symcall_i^{0}, w\inputendmarker, Z_0 \gamma i \myspl\! \gamma' \mydollar) \myonlymove{}^\ast (q' \xi, w\inputendmarker\backslash\mymem_i(\gamma), Z_0 \gamma i \!\myspr\! \gamma' \mydollar)
    \]
    holds, and no ID with a state of $Q_\myautomaton$ appears in the steps $\myonlymove{}^\ast$.\footnote{A stack representation $\cdots \!\myspr\! Z \cdots \mydollar$ with head-reversed arrow $\mathord{\myspr}$ equals $\cdots Z \myspl\! \cdots \mydollar$.  In the rightmost ID, $i \!\myspr\! \gamma' \mydollar$ means $i \myspl\! \mydollar$ if $\gamma' = \varepsilon$.}
\end{proposition}
\begin{proof}
    For the case $\mathord{\lbrack_i} \notin \gamma$, we have $\mymem_i(\gamma) = \varepsilon$ and thus $w\inputendmarker \backslash \mymem_i(\gamma) = w\inputendmarker$.  Hence, the following calculation verifies the claim:
    \[
        \begin{array}{ll}
            & (q' \xi \symcall_i^{0}, w\inputendmarker, Z_0 \gamma i \myspl\! \gamma' \mydollar) \\
            \myonlymove{} & (q' \xi \symcall_i^{1}, w\inputendmarker, Z_0 \gamma \myspl\! i \gamma' \mydollar) \\
        \myonlymove{}^\ast & (q' \xi \symcall_i^{1}, w\inputendmarker, Z_0 \myspl\! \gamma i \gamma' \mydollar) \\
        \myonlymove{} & (q' \xi \symret_i^{1}, w\inputendmarker,  Z_0 \!\myspr\! \gamma i \gamma' \mydollar) \\
        \myonlymove{}^\ast & (q' \xi \symret_i^{1}, w\inputendmarker, Z_0 \gamma i \myspl\! \gamma' \mydollar) \\
        \myonlymove{} & (q' \xi \symret_i^{0}, w\inputendmarker,  Z_0 \gamma i \!\myspr\! \gamma' \mydollar) \\
        \myonlymove{} & (q' \xi, w\inputendmarker, Z_0 \gamma i \!\myspr\! \gamma' \mydollar).
    \end{array}
    \]

    The case $\mathord{\lbrack_i} \in \gamma$ can be proved by induction on the number of number characters.  By assumption (b) and Claim~\ref{clm:mat}, $\gamma i$ is matching and $\gamma$ has the decomposition $\gamma = \gamma_1 \lbrack_i \gamma_2 \rbrack_i \gamma_3$ where $\gamma_2$ has neither $\mathord{\lbrack_i}$ and $\mathord{\rbrack_i}$, and $\gamma_3$ has no $\mathord{\lbrack_i}$.

    \begin{claim} \label{clm:inotgamma2}
        For the decomposition $\gamma = \gamma_1 \lbrack_i \gamma_2 \rbrack_i \gamma_3$ above, $i \notin \gamma_2$ holds.
    \end{claim}
    \begin{proof}
        In fact, $i \in \gamma_2$ leads to a decomposition $\gamma_2 = \gamma_{2,1} i \gamma_{2,2}$.  By assumption (b) and Claim~\ref{clm:mat}, $\gamma_1 \lbrack_i \gamma_{2,1} i$ is matching.  Thus, we obtain $\mathord{\rbrack_{i}} \in \gamma_{2,1}$, which contradicts $\mathord{\rbrack_i} \notin \gamma_2$.
    \end{proof}

    \begin{claim}
        Fix an arbitrary decomposition $\gamma_3 i = \gamma_{3,1} i \gamma_{3,2}$.
        Let $x = \zettaiti{\gamma_{3,2}}_i + 1 = t_{\zettaiti{\gamma_1 \lbrack_i \gamma_2 \rbrack_i \gamma_{3,1}},i}^{\gamma i}$ (see Definition~\ref{def:t} for $t$).\footnote{The notation $\zettaiti{\gamma}_i$ denotes the number of occurrences of $i$ in a string $\gamma$.}
        Then, $x \leq \tconst$ and 
        \[
            (q' \xi \symcall_i^{0}, w\inputendmarker, Z_0 \gamma i \myspl\! \gamma' \mydollar) \myonlymove{}^\ast (q' \xi \symcall_i^{x}, w\inputendmarker, Z_0 \gamma_1 \lbrack_i \gamma_2 \rbrack_i \gamma_{3,1} \myspl\! i \gamma_{3,2} \gamma' \mydollar).
        \]
        In particular, the decomposition in which $\gamma_{3,1}$ has no $i$ validates the following calculation with $x = \zettaiti{\gamma_3}_i + 1$:
        \[
            \begin{array}{ll}
                & (q' \xi \symcall_i^{0}, w\inputendmarker, Z_0 \gamma i \myspl\! \gamma' \mydollar) \\
                \myonlymove{}^\ast & (q' \xi \symcall_i^{x}, w\inputendmarker,  Z_0 \gamma_1 \lbrack_i \gamma_2 \rbrack_i \gamma_{3,1} \myspl\! i \gamma_{3,2} \gamma' \mydollar) \\
                \myonlymove{}^\ast & (q' \xi \symcall_i^{x}, w\inputendmarker, Z_0 \gamma_1 \lbrack_i \gamma_2 \rbrack_i \myspl\! \gamma_{3} \gamma' \mydollar) \\
                \myonlymove{} & (q' \xi \symcall_i^{x}, w\inputendmarker, Z_0 \gamma_1 \lbrack_i \gamma_2 \myspl \rbrack_i  \gamma_{3} \gamma' \mydollar) \\
                \myonlymove{}^\ast & (q' \xi \symcall_i^{x}, w\inputendmarker, Z_0 \gamma_1 \lbrack_i \,\myspl\! \gamma_2  \rbrack_i  \gamma_{3} \gamma' \mydollar).
            \end{array}
        \]
    \end{claim}
    \begin{proof}
        We prove by induction on $\zettaiti{\gamma_{3,2}}_i$.  The base case $\zettaiti{\gamma_{3,2}}_i = 0$ follows from
        \[
            (q' \xi \symcall_i^{0}, w\inputendmarker, Z_0 \gamma i \myspl\! \gamma' \mydollar) \myonlymove{} (q' \xi \symcall_i^{1}, w\inputendmarker, Z_0 \gamma \myspl\! i \gamma' \mydollar).
        \]
        Here, note that $\tconst \geq 1$ holds by assumption (b) and Proposition~\ref{prop:tconst}.\footnote{Alternatively, retake $\tconst$ of Proposition~\ref{prop:tconst} to be $\theta \geq 1$ if necessary.}
        Suppose that $\zettaiti{\gamma_{3,2}}_i \geq 1$ and the decomposition $\gamma_{3,2} = \gamma_{3,2,1} i \gamma_{3,2,2}$ where $\gamma_{3,2,1}$ has no $i$.  By the induction hypothesis, 
        \[
            \begin{array}{ll}
                & (q' \xi \symcall_i^{0}, w\inputendmarker, Z_0 \gamma i \myspl\! \gamma' \mydollar) \\
                \myonlymove{}^\ast & (q' \xi \symcall_i^{x'}, w, Z_0 \gamma_1 \lbrack_i \gamma_2 \rbrack_i \gamma_{3,1} i \gamma_{3,2,1} \myspl\! i \gamma_{3,2,2} \gamma' \mydollar) \\
                \myonlymove{}^\ast & (q' \xi \symcall_i^{x'}, w\inputendmarker, Z_0 \gamma_1 \lbrack_i \gamma_2 \rbrack_i \gamma_{3,1} i \myspl\! \gamma_{3,2,1} i \gamma_{3,2,2} \gamma' \mydollar)
            \end{array}
        \]
        where $x' = \zettaiti{\gamma_{3,2,2}}_i + 1 = t_{\zettaiti{\gamma_1 \lbrack_i \gamma_2 \rbrack_i \gamma_{3,1} i \gamma_{3,2,1}}, i}^{\gamma i}$.
        We claim that $x = x' + 1 \leq \tconst$.  In fact, $x' = t_{\zettaiti{\cdots \gamma_{3,1} i \gamma_{3,2,1}}, i}^{\gamma i} = t_{\zettaiti{\cdots \gamma_{3,1} i}, i}^{\gamma i} = t_{\zettaiti{\cdots \gamma_{3,1}}, i}^{\gamma i} - 1$ because $i \notin \gamma_{3,2,1}$ and from the definition of $t$.  Hence, $x = t_{\zettaiti{\cdots \gamma_{3,1}}, i}^{\gamma i} = x' + 1$.  Furthermore, by assumption (b), $\gamma i$ is a prefix of some $v \in \reflang(\alpha)$.  Therefore, 
        \[
            x' + 1 = t_{\zettaiti{\cdots \gamma_{3,1}}, i}^{\gamma i} \leq t_{\zettaiti{\cdots \gamma_{3,1}}, i}^{v} \leq \tconst.
        \]
        So, we can take the next calculation step with rule (4):
        \[
            \begin{array}{ll}
                & (q' \xi \symcall_i^{x'}, w\inputendmarker, Z_0 \gamma_1 \lbrack_i \gamma_2 \rbrack_i \gamma_{3,1} i \myspl\! \gamma_{3,2,1} i \gamma_{3,2,2} \gamma' \mydollar) \\
                \myonlymove{} & (q' \xi \symcall_i^{x}, w\inputendmarker, Z_0 \gamma_1 \lbrack_i \gamma_2 \rbrack_i \gamma_{3,1} \myspl\! i \gamma_{3,2} \gamma' \mydollar).
            \end{array}
        \]
        The proof has completed.
    \end{proof}

    Recall that $g: (\Sigma \mydisjointu \brkset)^\ast \to \Sigma^\ast$ denotes the free monoid homomorphism defined by $a \in \Sigma \mapsto a$ and $b \in \brkset \mapsto \varepsilon$ (Definition~\ref{def:deref}).
    When we write $\gamma_2 = z_0 k_1 z_1 \cdots k_m z_m$ where $k_1, \dots, k_m \in \mynat$ and none of $z_0, z_1, \dots, z_m$ has a number character, it can be easily checked that
    \[
        \mymem_i(\gamma) = g(z_0) \mymem_{k_1}(\gamma \lbrack_i z_0) g(z_1) \cdots \mymem_{k_m}(\gamma_1 \lbrack_i z_0 k_1 z_1 \cdots k_{m-1} z_{m-1}) g(z_m).
    \]
    We splice the calculation steps as follows:
    \[
        \begin{array}{ll}
            & (q' \xi \symcall_i^{x}, w\inputendmarker, Z_0 \gamma_1 \lbrack_i \,\myspl\! \gamma_2 \rbrack_i \gamma_3 i \gamma' \mydollar) \\
            \myonlymove{} & (q' \xi \symexec_i^{x}, w\inputendmarker,  Z_0 \gamma_1 \lbrack_i \myspr\! z_0 k_1 z_1 \cdots k_m z_m \rbrack_i \gamma_3 i \gamma' \mydollar) \\
            \myonlymove{}^\ast & (q' \xi \symexec_i^{x}, w\inputendmarker \backslash g(z_0), Z_0 \gamma_1 \lbrack_i z_0 k_1 \myspl\! z_1 \cdots k_m z_m \rbrack_i \gamma_3 i \gamma' \mydollar).
        \end{array}
    \]
    By assumption (b), $y = \gamma i \gamma'$ is a prefix of some $v \in \reflang(\alpha)$.  With (c), 
    \[
        \sconst \geq s_{\zettaiti{\gamma_1 \lbrack_i z_0},k_1}^{v} \geq s_{\zettaiti{\gamma_1 \lbrack_i z_0},k_1}^{y}
        \geq s_{\zettaiti{\gamma_1 \lbrack_i z_0 k_1},i}^{y} + 1 \geq s_{\zettaiti{\gamma},i}^{y} + 1 \geq \zettaiti{\xi} + 2.
    \]
    Therefore, $q \xi \symexec_i^{x} \symcall_{k_1}^{0} \in Q$.
    Note that the penultimate inequality follows from the definition of $s$ and the fact that $\gamma_2$ and $\gamma_3$ have no $\mathord{\lbrack_i}$.
    So, we can take the next calculation step with rule (7).  Moreover, by the induction hypothesis,\footnote{We do not need the induction hypothesis when $m = 0$, the base case of the induction.} we continue as follows:
    \[
        \begin{array}{ll}
            & (q' \xi \symexec_i^{x}, w\inputendmarker \backslash g(z_0), Z_0 \gamma_1 \lbrack_i z_0 k_1 \myspl\! z_1 \cdots k_m z_m \rbrack_i \gamma_3 i \gamma' \mydollar) \\
            \myonlymove{} & (q' \xi \symexec_i^{x} \symcall_{k_1}^{0}, w\inputendmarker \backslash g(z_0), Z_0 \gamma_1 \lbrack_i z_0 k_1 \myspl\! z_1 \cdots k_m z_m \rbrack_i \gamma_3 i \gamma' \mydollar) \\
            \myonlymove{}^\ast & (q' \xi \symexec_i^{x}, w\inputendmarker \backslash g(z_0) \mymem_{k_1}(\gamma_1 \lbrack_i z_0), Z_0 \gamma_1 \lbrack_i z_0 k_1 \!\myspr\! z_1 \cdots k_m z_m \rbrack_i \gamma_3 i \gamma' \mydollar) \\
            \myonlymove{}^\ast & (q' \xi \symexec_i^{x}, w\inputendmarker \backslash g(z_0) \mymem_{k_1}(\gamma_1 \lbrack_i z_0) g(z_1), Z_0 \gamma_1 \lbrack_i z_0 k_1 z_1 \!\myspr\! \cdots k_m z_m \rbrack_i \gamma_3 i \gamma' \mydollar).  
        \end{array}
    \]
    We can repeat this also for $k_2, \dots, k_m$.  Therefore, the following calculation
    \[ 
    \begin{array}{ll}
      & (q' \xi \symexec_i^{x}, w\inputendmarker \backslash g(z_0) \mymem_{k_1}(\gamma_1 \lbrack_i z_0) g(z_1), Z_0 \gamma_1 \lbrack_i z_0 k_1 z_1 \!\myspr\! \cdots k_m z_m \rbrack_i \gamma_3 i \gamma' \mydollar) \\
            \myonlymove{}^\ast & (q' \xi \symexec_i^{x}, w\inputendmarker \backslash \mymem_i(\gamma), Z_0 \gamma_1 \lbrack_i \gamma_2 \rbrack_i \,\myspl\! \gamma_3 i \gamma' \mydollar) \\
            \myonlymove{} & (q' \xi \symret_i^{x}, w\inputendmarker \backslash \mymem_i(\gamma), Z_0 \gamma_1 \lbrack_i \gamma_2 \rbrack_i \myspr\! \gamma_3 i \gamma' \mydollar) \\
            \myonlymove{}^\ast & (q' \xi \symret_i^{1}, w\inputendmarker \backslash \mymem_i(\gamma), Z_0 \gamma i \myspl\! \gamma' \mydollar) \\
            \myonlymove{} & (q' \xi \symret_i^{0}, w\inputendmarker \backslash \mymem_i(\gamma), Z_0 \gamma i \!\myspr\! \gamma' \mydollar) \\
            \myonlymove{} & (q' \xi, w\inputendmarker \backslash \mymem_i(\gamma), Z_0 \gamma i \!\myspr\! \gamma' \mydollar)
        \end{array}
    \]
    holds, and we complete the proof.
\end{proof}

\begin{proof}[Proof of Lemma~\ref{lem:irules}]
    We first show the ``$\Leftarrow$'' direction.  Let $\xi = \gamma' = \varepsilon$ in Proposition~\ref{prop:det}.  Fix $i$ with $q \myto{i}{\myautomaton} q'$ and splice $v$ as $q_0 \mytto{\gamma}{\myautomaton} q \myto{i}{\myautomaton} q'$.  From the condition of $\myautomaton$, $\gamma i$ is a prefix of some string in $\reflang(\alpha)$ and matching by Claim~\ref{clm:mat}.  The other requirements $q' \symcall_i^{0} \in Q$ and $s_{\zettaiti{\gamma},i}^{\gamma i} \geq 1$ are obvious.  With rule (3) and $w = \mymem_i(\gamma) w'$, we obtain the following calculation, where the steps $\myid{}^\ast$ has no ID with a state of $Q_\myautomaton$:
    \[
        \begin{array}{ll}
            & (q,w\inputendmarker, Z_0 \gamma \myspl\! \mydollar) \\
            \myid{(3)} & (q' \symcall_i^{0}, w\inputendmarker,  Z_0 \gamma i \myspl\! \mydollar) \\
            \myid{}^\ast & (q', w'\inputendmarker, Z_0 \gamma i \myspl\! \mydollar).
        \end{array}
    \]

    We next show the ``$\Rightarrow$'' direction.  The first application of the assumed calculation implies $q \myto{i}{\myautomaton} q'$.  Therefore, by the same argument discussed above, we obtain the deterministic calculation $(q' \symcall_i^{0}, w\inputendmarker, Z_0 \gamma i \myspl\! \mydollar) \myonlymove{}^\ast (q', w\inputendmarker\backslash\mymem_i(\gamma), Z_0 \gamma i \myspl\! \mydollar)$ with the condition of $\myonlymove{}^\ast$.  Thanks to the determinism and the condition, the two calculations must coincide and so must both ends.  This completes the proof.
\end{proof}

\begin{proof}[Proof of Theorem~\ref{thm:rewbscisnesa}]
        Take any $w \in L(\alpha)$ and fix $v = v_1 \cdots v_{\zettaiti{v}} \in \reflang(\alpha)$ with $\deref(v) = w$ and $v_1, \dots, v_{\zettaiti{v}} \in \nfaalph$.  Let $m = \zettaiti{v}$.  Also fix a run $q_0 \myto{v_1}{\myautomaton} q_1 \cdots \myto{v_m}{\myautomaton} q_m \in F$.
        Decompose $v = x_0 i_1 x_1 \cdots i_n x_n$ where $x_r \in (\Sigma \mydisjointu \brkset)^\ast$ and $i_r \in \mynat$.
        Let $y_{r-1} = x_0 \cdots i_{r-1} x_{r-1}$ for each $r \in \{ 1, \dots, n \}$.
        We can easily check that $w = \deref(v) = g(x_0) \mymem_{i_1}(y_1) g(x_1) \cdots \mymem_{i_n}(y_{n}) g(x_n)$.
        Hence, the ``$\Leftarrow$'' directions of Lemmas~\ref{lem:abrule} and \ref{lem:irules} weave the run into the calculation $(q_0, w\inputendmarker, Z_0 \myspl\! \mydollar) \myid{}^\ast (q_m, \inputendmarker, Z_0 v \myspl\! \mydollar)$ while maintaining the invariant that the stack is $Z_0 y_{r-1} \myspl\! \mydollar$ at the step guessing $i_{r}$.
        Conversely, take $w \in L(A_\alpha)$ and fix a calculation $C_0 \myid{} C_1 \cdots \myid{} C_m$ where $C_r = (p_r, w_r\inputendmarker, \beta_r \mydollar)$ that is accepting, namely $p_0 = q_0$, $w_0 = w$, $\beta_0 = Z_0 \myspl$, $p_m \in F$ and $w_m = \varepsilon$.
        Enumerate all the IDs $C_{\ell_0} = C_0, C_{\ell_1}, \dots, C_{\ell_n} = C_m$ in order whose state is in $Q_\myautomaton$.  For each $j < n$, the rule applied to $C_{\ell_j}$ is one of rules (1)--(3).
        Therefore, the ``$\Rightarrow$'' directions of Lemmas~\ref{lem:abrule} and \ref{lem:irules} again translate the calculation into the path $p_{\ell_0} \mytto{\gamma_{\ell_j}}{\myautomaton} p_{\ell_j}$ for $j = 0, \dots, n$ with the invariants that $\beta_{\ell_j} = Z_0 \gamma_{\ell_j} \myspl$ and $w = \deref(\gamma_{\ell_j}) w_{\ell_j}$.  Letting $j = n$ concludes the proof.
    \end{proof}

\section{Related Work}
\label{sec:related}

%
The seminal work for regular expressions with backreferences is by Aho~\cite{10.5555/114872.114877}.  
%
%
%
C\^{a}mpeanu et al.~formalized the semantics of rewbs using an ordered tree~\cite{campeanu2003formal}.
More recently, Schmid formalized rewbs with the notions of ref-languages and dereferencing functions, and proposed \emph{memory automata} (MFA), a class of automata equivalent to rewbs~\cite{schmid2016characterising}.  
Freydenberger and Schmid examined deterministic variants of rewbs and MFAs~\cite{freydenberger2019deterministic}.
%
%
%
As stated in the introduction, the relationship between the expressive power of rewbs and those of other well-known language classes such as CFL, NESL, SL, IL and CSL has been investigated~\cite{campeanu2003formal, berglund2023re, mfcs2023, ic2025}.

%
%
%
%

We want to refer to syntactic restrictions of rewbs related to this paper.
%
As mentioned in Definition~\ref{def:closed}, the closedness condition was originally proposed by C\^{a}mpeanu et al.~\cite{campeanu2003formal} and named \emph{no unassigned references} (NUR) by Chida and Terauchi~\cite{chida2022lookaheads}.  However, unlike us, these papers asserted the condition globally on the entire rewb and did not consider asserting the condition locally at subexpressions of the rewb.
Berglund and van~der~Merwe defined a \emph{non-circular} rewb, which excludes mutual references such as $((_1\bs 2 a)_1 (_2 \bs 1 \bs 1)_2)^\ast$~\cite{berglund2023re}.
Larsen introduced the notion of nested levels and showed a language $L_i$ expressible by a rewb at each nested level $i$ but not at any levels below $i$~\cite{larsen1998regular}.  Roughly, the nested level of a rewb represents the nesting depth of capturing groups and references.
Our recent work showed that every rewb without captured references describes an NESL and each $L_i$ of Larsen's hierarchy is also an NESL~\cite{mfcs2023,ic2025}.  

%
The use of a reference $\bs i$ inside a star has been considered as a bizarre syntax increasing the difficulty of a rewb.
Freydenberger and Holldack proposed a syntactic condition called \emph{variable star-free} (vstar-free)~\cite{freydenberger2018document}, which asserts that every subexpression immediately under a star must be a pure regular expression.
The closed-star condition proposed in our paper is similar to vstar-free in that it also asserts a condition on the subexpression immediately under a star.  However, our closed-star condition is strictly more lenient than vstar-freeness: it is easy to see that every vstar-free rewb is closed-star, but for instance, $((_1 a^\ast )_1 \bs 1)^\ast$ is closed-star but not vstar-free.
%
Freydenberger and Schmid investigated the computational complexity of some basic problems such as intersection-emptiness, inclusion and equivalence of vstar-free rewbs that additionally satisfy a certain form of determinism~\cite{freydenberger2019deterministic}. 

%
%

%
Next, we discuss PMCFLs and MCFLs.
Vijay-Shanker et al.~\cite{vijay1987characterizing} and Weir~\cite{weir1988characterizing} proposed a model called \emph{linear context-free rewriting systems} (LCFRS) and showed its semilinearity and polynomial time recognition.
Independently, Seki et al.~\cite{seki1991multiple} introduced PMCFG and MCFG and examined those expressive powers and properties.
In particular, they showed that MCFG is equivalent to LCFRS.
It is widely recognized that MCFG is a natural generalization of CFG and has several characterizations besides LCFRS, such as 
\emph{context-free hyper-graph grammar} (CFHG)~\cite{bauderon1987graph,habel1987some} and 
\emph{deterministic tree-walking transducer} (DTWT)~\cite{aho1969translations}.\footnote{The name seems to have been introduced by \cite{engelfriet1991cfhgdtwt} and the original paper~\cite{aho1969translations} used another name \emph{tree automaton}.}
%
The equivalence between CFHG and DTWT is shown by Engelfriet and Heyker~\cite{engelfriet1991cfhgdtwt} and that between LCFRS and DTWT is shown by Weir~\cite{weir1992lcfrsdtwt}.
More recently, Denkinger characterized $m$-MCFG by an automata model called \emph{$m$-restricted tree stack automata}~\cite{denkinger2016automata}.
Seki et al.~\cite{seki1991multiple} also proved a pumping lemma for $m$-MCFLs.  However, this pumping lemma is somewhat weaker than ordinary ones (see Theorem~\ref{thm:pumpMCFL}).
Kanazawa proved a pumping lemma of the usual form for a subclass of MCFLs called \emph{well-nested MCFLs}~\cite{kanazawa2009pumping}.
As remarked in the introduction, unary PMCFLs and MCFLs can be characterized by EDT0L systems and those of finite index~\cite{nishida2000grouped, kanazawa2010copying}.
In particular, Rozenberg and Vermeir investigated the properties of ET0L systems of finite index~\cite{rozenberg1978etol}.  One of their results is the equivalence of $\text{ET0L}_\text{FIN}$ and $\text{EDT0L}_\text{FIN}$.

\section{Conclusion}
\label{sec:conc}

In this paper, we have shown (1) that every rewb describes a unary-PMCFL (Theorem~\ref{thm:rewbpmcfg}), (2) however that there exists a rewb without captured references whose language is not an MCFL (Theorem~\ref{thm:notmcfg}).  As previously stated, this strictly improves the non-trivial upper bound on the expressive power of rewbs.  As a corollary, we also obtain that unary-PMCFL is not contained in SL and NESL is not contained in MCFL (Corollary~\ref{cor:pmcfgsaandmcfgnesa}).
Furthermore, we sought for a simple syntactic restriction of rewbs that produces a sufficient condition for the inclusion by the class of unary-MCFLs and presented the closed-star condition, which asserts each subexpression immediately under a star be closed (Definition~\ref{def:sc}).
%
We have shown that the condition provides an upper bound on the number of times a rewb references the same captured string, and we utilized this property to prove (3) that the language class of closed-star rewbs falls inside the intersection of the class of unary-MCFLs and NESLs (Theorems~\ref{thm:rewbscmcfg} and \ref{thm:rewbscisnesa}).
From (2) and (3), we obtained that closed-star rewbs have weaker expressive power than general rewbs (Corollary~\ref{cor:scsmaller}).
Nonetheless, the closed-star condition is a simple syntactic condition, and despite the simplicity, as we have shown in the paper, closed-star rewbs can express interesting non-context-free languages like $L(((_1 (a+b)^\ast )_1\, c\,\bs 1)^\ast)$ for instance.
Note that the inclusions (1) and (3) are strict because a unary-CFG (i.e., linear CFG) and an NESA can easily describe $\{a^n b^n \,|\, n \geq 0 \}$ but rewbs cannot~\cite{berglund2023re}.


\begin{figure}[htb] \centering
    \begin{tikzcd}[column sep={between origins,11em}, row sep={between origins,4.4em}]
        \text{REWB{\textbackslash}CR} \arrow[r, "\text{\cite{mfcs2023,ic2025}}"] \arrow[rd] &
        \text{NESL} \arrow[r] &
        \text{SL} \arrow[rd] &
        \\ 
        &
        \text{REWB} \arrow[ru, "/" marking, dash, "\text{\cite{mfcs2023,ic2025}}"] \arrow[rd, mydeepred, dashed, crossing over, pos=0.9, "\text{Th.\,\ref{thm:rewbpmcfg}}"] &
        \text{MCFL} \arrow[r, "/" marking, dash] \arrow[from=ld, crossing over] \arrow[from=llu, mydeepred, "/" marking, dash, dashed, crossing over] \arrow[from=llu, phantom, mydeepred, shift left=4, pos=0.95, "\text{\scriptsize Th.\,\ref{thm:notmcfg}}"] &
        \text{IL} \\ 
        \text{CS-REWB} \arrow[ruu, mydeepred, dashed, crossing over, pos=0.19, "\text{Th.\,\ref{thm:rewbscisnesa}}"] \arrow[ru, mydeepred, dashed] \arrow[r, mydeepred, dashed, swap, "\text{Th.\,\ref{thm:rewbscmcfg}}"] &
        \text{unary-MCFL} \arrow[ru] \arrow[d, equal] \arrow[r] &
        \text{unary-PMCFL} \arrow[ru, crossing over, swap] \arrow[d, equal] &
        \\ [-1.5em]
        &
        \text{EDT0L$_{\text{FIN}}$} &
        \text{EDT0L} &
    \end{tikzcd}
\end{figure}

    The diagram above depicts the selected relations among the classes mentioned in the paper.  Here, an arrow denotes a strict inclusion and a negated line denotes being incomparable.  A citation refers to the evidence.  Red dashed arrows and lines indicate novel results proved in this paper, where for each strict inclusion, we show for the first time the inclusion itself in addition to the fact that it is strict.  ``REWB{\textbackslash}CR'' stands for the language class of rewbs without captured references and ``CS-REWB'' for that of closed-star rewbs.  The full relations are available in \ref{app:table}.

A possible direction for future work is to further improve the non-trivial upper bounds on rewbs and closed-star rewbs by considering unidirectionality in the sense that concatenate functions in our PMCFG and MCFG construction do tail-appending only.
%
    Additionally, we would like to investigate open inclusions among various syntactic conditions of rewbs proposed by others and us, such as nested levels~\cite{larsen1998regular}, no captured references~\cite{mfcs2023,ic2025}, non-circular~\cite{berglund2023re}, NUR~\cite{campeanu2003formal,carle2009extended,chida2022lookaheads}, vstar-free~\cite{freydenberger2018document,freydenberger2019deterministic}, and closed-star.

\section{Acknowledgements}
\label{sec:acks}

We are grateful to Yasuhiko Minamide and Takayuki Miyazaki for their insightful comments on MCFLs and PMCFLs.
This work was supported by JSPS KAKENHI Grant Numbers JP25KJ2137, JP23K24826 and JP26K23826.


\bibliography{refs}

\appendix

\begin{landscape}
\section{A comprehensive table of the relations among the language classes}
\label{app:table}

\newcommand{\mytrue}{T}
\newcommand{\myfalse}{F}
\newcommand{\myemph}[1]{\cellcolor[gray]{0.9}\textbf{#1}}
\newcolumntype{I}{!{\vrule width 1.3pt}}
\newcolumntype{i}{!{\vrule width 1.15pt}}
\newcommand{\bhline}[1]{\noalign{\hrule height #1}}
\footnotesize
\begin{table}[htb!]
    \begin{tabular}{Icic|c|c|c|c|c|c|c|c|c|cI} \bhline{1.3pt}
        $\mathcal{L}_1 \backslash \mathcal{L}_2$ & IL  & SL & NESL & CFL & PMCFL & MCFL & U-PMCFL & U-MCFL & REWB & REWB{\textbackslash}CR & CS-REWB \\ \bhline{1.15pt}
        IL       & $=$ & \myfalse~\cite{aho1969nested,greibach1969checking} & \myfalse & \myfalse & \myfalse\textsuperscript{1} & \myfalse & \myfalse & \myfalse & \myfalse & \myfalse & \myfalse \\ \hline
        SL       & \mytrue & $=$ & \myfalse~\cite{ogden1969intercalation,greibach1969checking} & \myfalse & ? & \myemph{\myfalse} & \myfalse & \myfalse & \myfalse & \myfalse & \myfalse \\ \hline
        NESL     & \mytrue & \mytrue & $=$ & \myfalse\textsuperscript{2} & ? & \myemph{\myfalse~(\ref{cor:pmcfgsaandmcfgnesa})} & ? & \myemph{\myfalse} & \myfalse & \myfalse & \myfalse \\\hline
        CFL      & \mytrue & \mytrue & \myfalse~\cite{engelfriet1980stack} & $=$ & \mytrue & \mytrue & \myfalse~\cite{ehrenfeucht1977some} & \myfalse & \myfalse\textsuperscript{3} & \myfalse & \myfalse \\\hline
        PMCFL    & \myfalse& \myfalse & \myfalse & \myfalse & $=$ & \myfalse~\cite{seki1991multiple} & \myfalse & \myfalse & \myfalse & \myfalse & \myfalse \\ \hline
        MCFL     & \myfalse\textsuperscript{4} & \myfalse & \myfalse & \myfalse & \mytrue & $=$ & \myfalse & \myfalse & \myfalse & \myfalse & \myfalse \\ \hline
        U-PMCFL & \mytrue\textsuperscript{5} & \myemph{\myfalse~(\ref{cor:pmcfgsaandmcfgnesa})} & \myemph{\myfalse} & \myfalse & \mytrue & \myemph{\myfalse} & $=$ & \myfalse~\cite{rozenberg1978etol} & \myfalse & \myfalse & \myfalse \\ \hline
        U-MCFL & \mytrue & ? & ? & \myfalse~\cite{rozenberg1978etol} & \mytrue & \mytrue  & \mytrue & $=$ & \myfalse\textsuperscript{3} & \myfalse & \myfalse \\ \hline
        REWB & \mytrue~\cite{mfcs2023,ic2025} & \myfalse~\cite{mfcs2023,ic2025} & \myfalse & \myfalse & \myemph{\mytrue} & \myemph{\myfalse} & \myemph{\mytrue~(\ref{thm:rewbpmcfg})} & \myemph{\myfalse} & $=$ & \myfalse & \myemph{\myfalse} \\ \hline
        REWB{\textbackslash}CR & \mytrue & \mytrue & \mytrue~\cite{mfcs2023,ic2025} & \myfalse\textsuperscript{2} & \myemph{\mytrue} & \myemph{\myfalse~(\ref{thm:notmcfg})} & \myemph{\mytrue} & \myemph{\myfalse} & \mytrue & $=$ & \myemph{\myfalse} \\ \hline
        CS-REWB & \mytrue & \myemph{\mytrue} & \myemph{\mytrue~(\ref{thm:rewbscisnesa})} & \myfalse\textsuperscript{2} & \myemph{\mytrue} & \myemph{\mytrue} & \myemph{\mytrue} & \myemph{\mytrue~(\ref{thm:rewbscmcfg})} & \mytrue & ? & $=$ \\ \bhline{1.3pt}
    \end{tabular}
    \caption{Summary of whether $\mathcal{L}_1 \subseteq \mathcal{L}_2$ for most of the language classes considered in this paper.  \mytrue{} means $\mathcal{L}_1 \subseteq \mathcal{L}_2$, \myfalse{} means $\mathcal{L}_1 \nsubseteq \mathcal{L}_2$, and ? means that the relation seems open.  U-PMCFL and U-MCFL stand for the classes of unary-PMCFLs and unary-MCFLs, respectively.  The relations proved for the first time by this paper are highlighted and each parenthesized number refers to the corresponding theorem or corollary in this paper (highlighted cells without parenthesized numbers are novel results that follow from other novel results of the paper).  A citation refers to the evidence and a superscript refers to the item of the list below that explains why the relation holds.}
\end{table}
    \begin{enumerate}
        \item Nishida and Seki showed an ET0L language that is not in PMCFL~\cite{nishida2000grouped}.  ET0L $\subseteq$ IL is obvious (see, e.g.,~\cite{salomaa1974parallelism}).
            Note that the fact that this inclusion is strict (i.e., $\text{ET0L} \subsetneq \text{IL}$) was firstly shown by Ehrenfeucht et al.~\cite{ehrenfeucht1976relationship}.
        \item The copy language $\{ ww \,|\, w \in \{a, b\}^\ast \}$ is a witness.
        \item The language $\{ a^nb^n \,|\, n \geq 0 \}$ is a witness~\cite{berglund2023re}.
        \item It follows from the copying theorem for IL~\cite{engelfriet1976copying} and the existence of a CFL that is not an EDT0L~\cite{ehrenfeucht1977some} (see also \cite{kanazawa2010copying}). 
        \item $\text{Unary-PMCFL} = \text{EDT0L} \subseteq \text{ET0L} \subseteq \text{IL}$ holds, where the first equation follows from \cite{nishida2000grouped} as remarked in the introduction,
            and the latter two inclusions are obvious as mentioned in 1.
    \end{enumerate}
\end{landscape}

\end{document}